\documentclass[a4paper]{LMCS}
\usepackage{ltl_sat}
\usepackage{enumerate,hyperref}

\reporttrue

\def\doi{5 (1:1) 2009}
\lmcsheading%
{\doi}
{1--21}
{}
{}
{Mar.~28, 2008}
{Jan.~26, 2009}
{}   

\begin{document}
  \title[%
  The Complexity of Generalized Satisfiability for LTL%
]{%
  The Complexity of Generalized Satisfiability \\ for Linear Temporal Logic%
}

\author[M.~Bauland]{Michael Bauland\rsuper a}
\address{{\lsuper a}Knipp GmbH, Martin-Schmei\ss er-Weg 9, 44227 Dortmund, Germany}
\email{Michael.Bauland@knipp.de}


\author[T.~Schneider]{Thomas Schneider\rsuper b}
\address{{\lsuper b}School of Computer Science, University of Manchester, Oxford Road, Manchester M13 9PL, UK}
\email{schneider@cs.man.ac.uk}

\author[H.~Schnoor]{Henning Schnoor\rsuper c}
\address{{\lsuper c}Inst.\ f\"ur Informatik, Christian-Albrechts-Universit\"at zu Kiel, 24098 Kiel, Germany}
\email{schnoor@ti.informatik.uni-kiel.de}
\thanks{Supported by the Postdoc Programme of the German Academic Exchange Service (DAAD)}

\author[I.~Schnoor]{Ilka Schnoor\rsuper d}
\address{{\lsuper d}Inst.\ f\"{u}r Theoretische Informatik, Universit\"{a}t zu L\"{u}beck, Ratzeburger Allee 160, 23538 L\"{u}beck, Germany}
\email{schnoor@tcs.uni-luebeck.de}

\author[H.~Vollmer]{Heribert Vollmer\rsuper e}
\address{{\lsuper e}Inst.\ f\"ur Theoretische Informatik, Universit\"{a}t Hannover, Appelstr. 4, 30167 Hannover, Germany}
\email{vollmer@thi.uni-hannover.de}
\thanks{Supported in part by DFG VO 630/6-1.}

\titlecomment{{\lsuper*}This article extends the conference contribution \cite{bss+07} with full proofs of all lemmata and theorems.}

\begin{abstract}
  In a seminal paper from 1985, Sistla and Clarke showed
  that satisfiability for Linear Temporal Logic (LTL) is either \NP-complete
  or \PSPACE-complete, depending on the set of temporal operators used%
  .
  If, in contrast, the set of propositional operators is restricted, the complexity may decrease.
  This paper undertakes a systematic study of satisfiability for LTL formulae over restricted sets
  of propositional and temporal operators. Since every propositional operator corresponds to a
  Boolean function, there exist infinitely many propositional operators. In order to systematically
  cover all possible sets of them, we use Post's lattice. With its help, we determine the
  computational complexity of LTL satisfiability for all combinations of temporal operators and
  all but two classes of propositional functions. Each of these infinitely many problems is shown
  to be either \PSPACE-complete, \NP-complete, or in \PTIME.
\end{abstract}

\maketitle
\vfill

\section{Introduction}

  \newlength{\abo}\setlength{\abo}{-3pt}
  \newlength{\abu}\setlength{\abu}{0pt}
  \newlength{\Abu}\setlength{\Abu}{4pt}

  \newcommand{\Stab}{\rule{0pt}{14pt}}
  \newcommand{\stab}{\rule{0pt}{12pt}}


%

  \noindent
  \emph{Linear Temporal Logic (LTL)}~ 
  was introduced by Pnueli in~\cite{pnu77} as a formalism
  for reasoning about the properties and the behaviors of parallel programs and concurrent systems,
  and has widely been used
  for these purposes. Because of the need to perform reasoning tasks---such as deciding
  satisfiability, validity, or truth in a structure generated by binary relations---in an
  automated manner, their decidability and computational complexity is an important issue.

  It is known that in the case of full LTL with the operators \F\ (eventually), \G\ (invariantly),
  \X\ (next-time), \U\ (until), and \S\ (since), satisfiability and determination of truth
  are \PSPACE-complete~\cite{sicl85}.
  Restricting the set of temporal operators leads to \NP-completeness in some cases~\cite{sicl85}.
  These results imply that reasoning with LTL is difficult in terms of computational complexity.

  This raises the question under which restrictions the complexity of these problems decreases.
  Contrary to classical modal logics, there does not seem to be a natural way to modify the semantics of LTL and obtain decision problems with lower complexity.
  However, there are several possible constraints that can be posed on the syntax.
  One possibility is to restrict the set of temporal operators, which has been done 
  exhaustively in~\cite{sicl85,mar04}.

  Another constraint is to allow only a certain ``degree of propositionality'' in the language,
  \ie, to restrict the set of allowed propositional operators. Every propositional operator
  represents a Boolean function---\eg, the operator $\wedge$ (\AND) corresponds to the
  binary function whose value is 1 if and only if both arguments have value 1. There are
  infinitely many Boolean functions and hence an infinite number of propositional operators.

  We will consider propositional restrictions in a systematic way, achieving a
  complete classification of the complexity of the reasoning problems for LTL.
  Not only will this reveal all cases in this framework where satisfiability is tractable. It will also
  provide a better insight into the sources of hardness by explicitly stating the combinations
  of temporal and propositional operators that lead to \NP- or \PSPACE-hard fragments.
  In addition, the ``sources of hardness'' will be identified whenever a proof technique
  is not transferable from an easy to a hard fragment.

  \par\medskip\noindent
  \emph{Related work.}~ The complexity of model-checking and
  satisfiability problems for several syntactic restrictions of LTL
  fragments has been determined in the literature: In
  \cite{sicl85,mar04}, temporal operators and the use of negation have
  been restricted; these fragments have been shown to be \NP- or
  \PSPACE-complete.  In \cite{ds02}, temporal operators, their
  nesting, and the number of atomic propositions have been restricted;
  these fragments have been shown to be tractable or \NP-complete.
  Furthermore, due to \cite{CL93,DFR00}, the restriction to Horn
  formulae does not decrease the complexity of satisfiability for LTL.
  As for related logics, the complexity of satisfiability has been
  shown in \cite{ees90} to be tractable or \NP-complete for three
  fragments of CTL (computation tree logic) with temporal and
  propositional restrictions. In \cite{hal95}, satisfiability for
  multimodal logics has been investigated systematically, bounding
  the depth of modal operators and the number of atomic
  propositions. In \cite{hem01}, it was shown that satisfiability
  for modal logic over linear frames drops from \NP-complete to
  tractable if propositional operators are restricted to conjunction
  and atomic negation.

  The effect of propositional restrictions on the complexity of the satisfiability problem
  was first considered \emph{systematically} 
  by Lewis for the case of classical propositional logic in~\cite{lew79}.
  He established a dichotomy---depending on the set of propositional operators,
  satisfiability is either \NP-complete or decidable in polynomial time. In the case of modal
  propositional logic, a trichotomy has been achieved in~\cite{bhss06}: modal satisfiability
  is \PSPACE-complete, \CONP-complete, or in \PTIME. That complete classification in terms of
  restrictions on the propositional operators follows the structure of Post's
  lattice of closed sets of Boolean functions~\cite{pos41}.

  \par\medskip\noindent
  \emph{Our contribution.}~
  This paper analyzes the same systematic propositional restrictions for LTL, 
  and combines them
  with restrictions on the temporal operators. Using Post's lattice,
  we examine the satisfiability problem for every possible fragment of
  LTL determined by an arbitrary set of propositional operators
  \textit{and} any subset of the five temporal operators listed
  above. We determine the computational complexity of these problems,
  except for one case---where only propositional operators
  based on the binary \XOR\ function (and, perhaps, constants) are
  allowed. We show that all remaining cases are either
  \PSPACE-complete, \NP-complete, or in \PTIME.

  It is not the aim of this paper to focus on particular propositional
  restrictions that are motivated by certain applications. We prefer
  to give a classification as complete as possible which allows to
  choose a fragment that is appropriate, in terms of expressivity and
  tractability, for any given application.
  Applications of syntactically restricted fragments of temporal
  logics can be found, for example, in the study of cryptographic
  protocols: In \cite{low08}, Gavin Lowe restricts the application of
  negation and temporal operators to obtain practical verification
  algorithms.

  Among our results, we exhibit cases with non-trivial tractability as well as
  the smallest possible sets of propositional and temporal operators that already lead
  to \NP-completeness or \PSPACE-completeness, respectively. Examples for the first
  group are cases in which only the unary \NOT\ function, or only monotone functions are
  allowed, but there is no restriction on the temporal operators. As for the second group,
  if only the binary function $f$ with $f(x,y) = (x \wedge \overline{y})$ is permitted,
  then satisfiability is \NP-complete already in the case of propositional
  logic~\cite{lew79}. Our results show that the presence of the same function $f$
  separates the tractable languages from the \NP-complete and \PSPACE-complete ones,
  depending on the set of temporal operators used. According to this, minimal sets of temporal operators
  leading to \PSPACE-completeness together with $f$ are, for example, $\{\U\}$ and
  $\{\F,\X\}$.

  The technically most involved proof is that of \PSPACE-hardness for the language
  with only the temporal operator \S\ and the boolean operator $f$
  (Theorem~\ref{theorem:SAT(S;BF) SAT(S|U;S1) PSPACE-h}). The difficulty lies in
  simulating the quantifier tree of a Quantified Boolean Formula (QBF) in a linear structure.

  Our results are summarized in Table~\ref{tab:results_sat}. The first column contains
  the sets of propositional operators, with the terminology taken from Definition \ref{def:BFs}.
  The second column shows the classification of
  classical propositional logic as known from~\cite{lew79} and~\cite{coo71a}.
  The last line in column 3 and 4 is largely due to~\cite{sicl85}. All other
  entries are the main results of this paper.
  The only open case appears in the third line and is discussed in the
  Conclusion. Note that the case distinction
  also covers all clones which are not mentioned in the present paper.
  \begin{table}
    \centering
    \begin{small}
      \begin{tabular}{l|c|cc}
        \stab \hspace*{\fill} set of temporal operators & $\emptyset$ & $\{\F\}$, $\{\G\}$,   & any other         \\[2pt]
        set of propositional operators                  &             & $\{\F,\G\}$, $\{\X\}$ & combination       \\
        \hline
        \Stab all operators 1-reproducing or self-dual  & trivial     & trivial               & trivial           \\
        \stab only negation or all operators monotone   & in \PTIME   & in \PTIME             & in \PTIME         \\
        \stab all operators linear                      & in \PTIME   & ?                     & ?                 \\
        \stab $x \wedge \neg y$ is expressible          & \NP-c.\     & \NP-c.\               & \PSPACE-c.\       \\
        \stab all Boolean functions                     & \NP-c.\     & \NP-c.\               & \PSPACE-c.\
      \end{tabular}

    \end{small}
    \par\bigskip
    \caption{%
      Complexity results for satisfiability. The entries ``trivial'' denote cases in which a given formula
      is always satisfiable. The abbreviation ``c.'' stands for ``complete.'' Question marks stand for open questions.%
    }
    \label{tab:results_sat}
  \end{table}

\section{Preliminaries}

  A \emph{Boolean function} or \emph{Boolean operator} is a function
  $f:\{0,1\}^n\rightarrow\{0,1\}$. We can identify an $n$-ary
  propositional connector $c$ with the $n$-ary Boolean operator $f$
  defined by: $f(a_1,\dots,a_n)=1$ if and only if the formula
  $c(x_1,\dots,x_n)$ becomes true when assigning $a_i$ to $x_i$ for
  all $1\leq i\leq n$. Additionally to propositional connectors we use
  the unary temporal operators $\X$ (next-time), $\F$ (eventually),
  $\G$ (invariantly) and the binary temporal operators $\U$ (until),
  and $\S$ (since).

  Let $B$ be a finite set of Boolean functions and $M$ be a set of
  temporal operators. A \emph{temporal $B$-formula over $M$} is a
  formula $\varphi$ that is built from variables, propositional
  connectors from $B$, and temporal operators from $M$.  More
  formally, a temporal $B$-formula over $M$ is either a propositional
  variable or of the form $f(\varphi_1,\dots,\varphi_n)$ or
  $g(\varphi_1,\dots,\varphi_m)$, where $\varphi_i$ are temporal
  $B$-formulae over $M$, $f$ is an $n$-ary propositional operator from
  $B$ and $g$ is an $m$-ary temporal operator from $M$. In
  \cite{sicl85}, complexity results for formulae using the temporal
  operators \F, \G, \X\ (unary), and \U, \S\ (binary) were
  presented. We extend these results to temporal $B$-formulae over
  subsets of those temporal operators.  The set of variables appearing
  in $\varphi$ is denoted by $V_{\varphi}.$ If $M=\{\X,\F,\G,\U,\S\}$
  we call $\varphi$ a \emph{temporal $B$-formula}, and if
  $M=\emptyset$ we call $\varphi$ a \emph{propositional $B$-formula}
  or simply a \emph{$B$-formula}. The set of all temporal $B$-formulae
  over $M$ is denoted by $\LL{M,B}.$

  A model in linear temporal logic is a linear structure of states,
  which intuitively can be seen as different points of time, with
  propositional assignments. Formally a \emph{structure} $S=(s,V,\xi)$
  consists of an infinite sequence $s=(s_i)_{i\in\mathbb{N}}$ of
  distinct states, a set of variables $V$, and a function
  $\xi:\{s_i\mid i\in\mathbb{N}\}\rightarrow 2^V$ which induces a
  propositional assignment of $V$ for each state. 
  be a structure and $\varphi$ a temporal $\{\wedge,\neg\}$-formula
  over $\{\X,\U,\S\}$ with variables from $V$. We define what it means
  that \emph{$S$ satisfies $\varphi$ in $s_i$} ($S,s_i\vDash\varphi$):
  For a temporal $\{\wedge,\neg\}$-formula over $\{\X,\U,\S\}$ with
  variables from $V$ we define what it means that \emph{$S$ satisfies
  $\varphi$ in $s_i$} ($S,s_i\vDash\varphi$): let $\varphi_1$ and
  $\varphi_2$ be temporal $\{\wedge,\neg\}$-formulae over
  $\{\X,\U,\S\}$ and $x\in V$ a variable.

  \begin{center}
    \begin{tabular}{lll}
      $S,s_i\vDash x$                        & if and only if & $x\in\xi(s_i)$,                                         \\
      $S,s_i\vDash \varphi_1\wedge\varphi_2$ & if and only if & $S,s_i\vDash \varphi_1$ and $S,s_i\vDash \varphi_2$,    \\
      $S,s_i\vDash \neg\varphi_1$            & if and only if & $S,s_i\nvDash\varphi_1$,                                \\
      $S,s_i\vDash \X\varphi_1$              & if and only if & $S,s_{i+1}\vDash\varphi_1$,                             \\
      $S,s_i\vDash \varphi_1\U\varphi_2$     & if and only if & there is a $k\geq i$ such that $S,s_k\vDash \varphi_2$, \\
                                             &                & and for every $i\leq j<k$,~ $S,s_j\vDash\varphi_1$,     \\
      $S,s_i\vDash \varphi_1\S\varphi_2$     & if and only if & there is a $k\leq i$ such that $S,s_k\vDash \varphi_2$, \\
                                             &                & and for every $k<j\leq i$,~ $S,s_j\vDash\varphi_1$.
    \end{tabular}
  \end{center}

  The remaining temporal operators are interpreted as abbreviations:
  $\F\varphi=\mathop{true}\U\varphi$ and
  $\G\varphi=\neg\F\neg\varphi.$ Therefore and since every Boolean
  operator can be composed from $\wedge$ and $\neg$, the above
  definition generalizes to temporal $B$-formulae for arbitrary sets
  $B$ of Boolean operators.

  A temporal $B$-formula $\varphi$ over $M$ is \emph{satisfiable} if
  there exists a structure $S$ such that $S,s_i\vDash\varphi$ for some
  state $s_i$ from $S$. Furthermore, $\varphi$ is called \emph{valid}
  if, for all structures $S$ and all states $s_i$ from $S$, it holds
  that $S,s_i\vDash\varphi$. We will consider the following problems:
  Let $B$ be a finite set of Boolean functions and $M$ a set of
  temporal operators. Then \tsat{M,B} is the problem to decide whether
  a given temporal $B$-formula over $M$ is satisfiable.
  In the literature, another notion of satisfiability is sometimes
  considered, where we ask if a formula can be satisfied at the first
  state in a structure. It is easy to see that, in terms of
  computational complexity, this does not make a difference for our
  problems as long as the considered fragment does not contain the
  temporal operator $\S$. For this paper, we only study the
  satisfiability problem as defined above.

Sistla and Clarke analyzed the satisfiability problem for temporal
$\{\wedge,\vee,\neg\}$-formulae over some sets of temporal operators,
see Theorem \ref{theorem:sicl85}. Note that, due to de Morgan's laws,
there is no significant difference between the sets
$\{\wedge,\vee,\neg\}$ and $\{\wedge,\neg\}$ of Boolean operators. For
convenience, we will therefore prefer the former denotation to the latter
when stating results. Furthermore, the original proof of Theorem
\ref{theorem:sicl85} explicitly uses the operator $\vee$.

  \begin{theorem}[\cite{sicl85}]\label{theorem:sicl85}
    ~\par
    \begin{enumerate}[\em(1)]
      \item\label{theorem:sicl85:SAT(F;BF) NP-c}
            $\tsat{\{\F\},\{\wedge,\vee,\neg\}}$ is \NP-complete.
      \item\label{theorem:sicl85:SAT(F,X|U|U,S,X;BF) PSPACE-c}
            $\tsat{\{\F,\X\},\{\wedge,\vee,\neg\}}$,
            $\tsat{\{\U\},\{\wedge,\vee,\neg\}}$, and
            $\tsat{\{\U,\S,\X\},\{\wedge,\vee,\neg\}}$ are \PSPACE-complete.
    \end{enumerate}
  \end{theorem}

  Since there are infinitely many finite sets of Boolean functions, we
  introduce some algebraic tools to classify the complexity of the
  infinitely many arising satisfiability problems. We denote with
  \proj{n}{k} the $n$-ary projection to the $k$-th variable, \ie,
  $\proj{n}{k}(x_1,\dots,x_n)=x_k$, and with \const{n}{a} the $n$-ary
  constant function defined by $\const{n}{a}(x_1,\dots,x_n)=a$. For
  $\const{1}{1}(x)$ and $\const{1}{0}(x)$ we simply write 1 and 0. A
  set $C$ of Boolean functions is called a \emph{clone} if it is
  closed under superposition, which means $C$ contains all projections
  and $C$ is closed under arbitrary composition \cite{pip97b}. For a
  set $B$ of Boolean functions we denote with $\clone{B}$ the smallest
  clone containing $B$ and call $B$ a \emph{base} for $\clone{B}$. In
  \cite{pos41} Post classified the lattice of all clones 
  Figure~\ref{Lattice}) and found a finite base for each clone.

  We now define some properties of Boolean functions, where $\oplus$ denotes the binary exclusive or. 
  \begin{definition}
  \label{def:BFs}
    Let $f$ be an $n$-ary Boolean function.
    \begin{enumerate}[$\bullet$]
      \item $f$ is 1\emph{-reproducing} if $f(1,\dots,1)=1$.
      \item $f$ is \emph{monotone} if $a_1\leq b_1,\dots,a_n\leq b_n$ implies $f(a_1,\dots,a_n)\leq f(b_1,\dots,b_n)$.
      \item $f$ is 1\emph{-separating} if there exists an $i\in\{1,\dots,n\}$ such that $f(a_1,\dots,a_n)=1$ implies $a_i=1$.
      \item $f$ is \emph{self-dual} if $f\equiv\dual(f)$, where $\dual(f)(x_1,\dots,x_n)= \neg f(\neg x_1,\dots,\neg x_n)$.
      \item $f$ is \emph{linear} if $f\equiv x_1\oplus\dots\oplus x_n\oplus c$ for a constant $c\in\set{0,1}$ and variables $x_1,\dots,x_n$.
    \end{enumerate}
  \end{definition}

  \begin{figure}
    \begin{center}
    \ifpdf
      \includegraphics{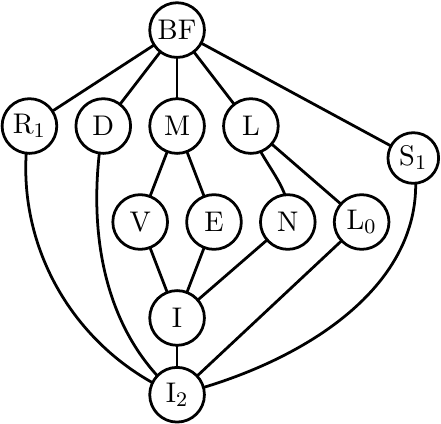}
    \else
      \includegraphics{Graphics/clones_restricted.1}
    \fi
    \caption{Graph of some closed classes of Boolean functions}
    \label{Lattice}
    \end{center}
  \end{figure}
  In Table~\ref{baselist} we define those clones that are essential for this paper plus four basic ones,
  and give Post's bases \cite{pos41} for them. The inclusions between them are given in Figure~\ref{Lattice}.
  The definitions of all clones as well as the full inclusion graph can be found, for example,
  in~\cite{bcrv03}.

%

  \def\Hline{\noalign{\hrule height.8pt}}
  \begin{table}
    \begin{scriptsize}
    \begin{center}
    \begin{tabular}{lll}
    \Hline Name & Definition & Base \\ 
    \Hline $\cBF$ & All Boolean functions & $\{\vee,\wedge,\neg\}$ \\ 
    \hline $\cR_1$ & $\{f \in \mathtext{BF} \mid f$ is $1$-reproducing $\}$ & $\{\vee,\ifff\}$ \\ 
    \hline $\cM$ & $\{f \in \mathtext{BF} \mid f$ is monotone $\}$ & $\{\vee,\wedge,0,1\}$ \\ 
    \hline $\cS_1$ & $\{f \in \mathtext{BF} \mid f$ is $1$-separating $\}$ & $\{x \wedge \overline{y}\}$ \\ 
    \hline $\cD$ & $\{f \mid f \text{ is self-dual}\}$ & $\{x\overline{y} \vee x\overline{z} \vee (\overline{y} \wedge \overline{z})\}$ \\ 
    \hline $\cL$ & $\{f \mid f\text{ is linear}\}$ & $\{\oplus,1\}$  \\ 
    \hline $\cL_0$ & $\clone{\{\oplus\}}$ & $\{\oplus\}$  \\
    \hline $\cV$ & $\{f \mid \text{There is a formula of the form } c_0 \vee c_1x_1 \vee \dots \vee c_nx_n$ & $\{\vee,1,0\}$  \\  &
    $\text{ such that } c_i \text{ are constants for } 1\leq i \leq n \text{ that describes } f\}$ & \\
    \hline $\cE$ & $\{f \mid \text{There is a formula of the form } c_0 \wedge(c_1\vee x_1)\wedge\dots\wedge(c_n\vee x_n)$ &
    $\{\wedge,1,0\}$ \\ & $\text{ such that } c_i \text{ are constants for } 1\leq i \leq n \text{ that describes } f\}$  &  \\
   \hline $\cN$ & $\{f \mid f\text{ depends on at most one variable}\}$  & $\{\neg,1,0\}$ \\
    \hline $\cI$ & $\{f \mid f\text{ is a projection or constant}\}$ & $\{0,1\}$ \\ 
    \hline $\cI_2$ & $\{f \mid f\text{ is a projection}\}$  & $\,\emptyset$ \\ 
    \Hline 
    \end{tabular}

    \end{center}

    \medskip
    \end{scriptsize}

    \caption{List of some closed classes of Boolean functions with bases}
    \label{baselist}
  \end{table}

  There is a strong connection between propositional formulae and Post's lattice.
  If we interpret propositional formulae as Boolean functions, it is
  obvious that $[B]$ includes exactly those functions that can be
  represented by $B$-formulae. This connection has been used various
  times to classify the complexity of problems related to
  propositional formulae: For example, Lewis presented a dichotomy for
  the satisfiability problem for propositional $B$-formulae:
  $\tsat{\emptyset,B}$ is \NP-complete if $\cS_1 \subseteq\clone{B}$,
  and solvable in \PTIME\ otherwise \cite{lew79}.

  Post's lattice was applied for the equivalence problem \cite{rei01},
  counting \cite{rewa99-dt} and finding minimal \cite{revo03}
  solutions, and learnability \cite{dal00} for Boolean formulae. The
  technique has been used in non-classical logic as well: Bauland et
  al. achieved a trichotomy in the context of modal logic, which says
  that the satisfiability problem for modal formulae is, depending on
  the allowed propositional connectives, \PSPACE-complete,
  \CONP-complete, or solvable in \PTIME\ \cite{bhss06}. For the
  inference problem for propositional circumscription, Nordh presented
  another trichotomy theorem \cite{nor05}.
 
  An important tool in restricting the length of the resulting formula
  in many of our reductions is the following lemma.  It shows that for
  certain sets $B$, there are always short formulae representing the
  functions \AND, \OR, or \NOT, respectively.  Point (2) and (3)
  follow directly from the proofs in \cite{lew79}, point (1) is
  Lemma~3.3 from \cite{sch05}.
\eject\vfill

  \begin{lemma}\label{lemma:short formulae}\hfill
    \begin{enumerate}[\em(1)]
    \item
      Let $B$ be a finite set of Boolean functions such that $V\subseteq\clone{B}\subseteq\mathtext{M}$
      ($E\subseteq\clone{B}\subseteq\mathtext{M}$, resp.). Then there exists a $B$-formula $f(x,y)$ such that $f$ represents
      $x\vee y$ ($x\wedge y$, resp.) and each of the variables $x$ and $y$ occurs exactly once in $f(x,y)$.
    \item
      Let $B$ be a finite set of Boolean functions such that $\clone B=\rm{BF}$. Then there are $B$-formulae $f(x,y)$ and
      $g(x,y)$ such that $f$ represents $x\vee y$, $g$ represents $x\wedge y$, and both variables occur in each of these
      formulae exactly once.
    \item
      Let $B$ be a finite set of Boolean functions such that $\mathtext{N}\subseteq\clone{B}$. Then there is a $B$-formula
      $f(x)$ such that $f$ represents $\neg x$ and the variable $x$ occurs in $f$ only once.
    \end{enumerate}
  \end{lemma}

\section{Results}

Our proofs for most of the upper complexity bounds will rely on similar
ideas as the ones in \cite{bhss06}, which are extensions of the proof
techniques for the polynomial time results in \cite{lew79}. However,
the proof of our polynomial time result for formulae using the
exclusive or (Theorem \ref{theorem:SAT(X;L) in P}) will be unrelated to the
positive cases for XOR in the mentioned papers.

The proofs for hardness results will use different techniques. Hardness
proofs for unimodal logics usually work in embedding a tree-like
structure directly into a tree-like model for modal
formulae. Naturally, this approach does not work with LTL which speaks
about linear models.  Hence, in the proof of Theorem
\ref{theorem:SAT(S;BF) SAT(S|U;S1) PSPACE-h}, we will encode a tree-like
structure into a linear one, and most of the complexity of the proof
will come from the need to enforce a tree-like behavior of linear models.

  \subsection{Hard cases}

  The following lemma gives our general upper bounds for various combinations of temporal operators.
  It establishes that the known upper complexity bounds for the case
  where only the propositional operators \AND, \OR, and negation are
  allowed to appear in the formulae still hold for the more general
  cases that we consider. This does not follow trivially, since there
  is no obvious strategy that converts every $B$-formula into a
  formula using only the standard connectives without leading to an
  exponential increase in formula length. The issues here are similar
  to the ``succinctness gap'' between the logics LTL+Past and LTL
  discussed in~\cite{mar04}. The proof of Parts (1) and (2) of the
  following lemma is a variation of the proof for Theorem 3.4 in
  \cite{bhss06}, where, using a similar reduction, an analogous result
  for circuits was proved.

  \begin{lemma}\label{lemma:PSPACE_ub_all}
  Let $B$ be a finite set of Boolean functions. Then the following holds:
  \begin{enumerate}[\em(1)]
   \item If $M \subseteq \set{\F,\G,\U,\S,\X}$, then $\tsat{M,B}$ is in \PSPACE,
   \item if $M \subseteq \set{\F,\G}$, then $\tsat{M,B}$ is in \NP, and
   \item if $M \subseteq \set{\X}$, then $\tsat{M,B}$ is also in \NP.
  \end{enumerate}
  \end{lemma}

  \proof
    For (1), we will show that $\tsat{M,B} \redlogm \tsat{\set{\U,\S,\X},\set{\wedge,\vee,\neg}},$ and
    for (2), we will show that $\tsat{M,B} \redlogm \tsat{\set{\F},\set{\wedge,\vee,\neg}}.$ The complexity result for these cases then follows from Theorem~\ref{theorem:sicl85}.
    \ifreport
    \else
      The proof for case (3) is omitted and given in \cite{bsssv06}.
    \fi

    The construction for (1) and (2) is nearly identical: Let $\varphi$ be a formula with arbitrary temporal operators and Boolean functions from $B$.
    We recursively transform the formula to a new formula using only the Boolean operators $\wedge$, $\vee$, and $\neg$, and the temporal
    operators \U, \S, and \X\ for the first case
    and the temporal operator \F\ for the second case\ifreport s\fi.
    For this we construct several formulae, which will be connected via conjunction. Let $k$ be the number of subformulae of $\varphi$.
    Accordingly let $\varphi_1, \dots, \varphi_k$ be those subformulae with $\varphi=\varphi_1$. Let $x_1,\dots,x_k$ be new
    variables, \ie, distinct from the input variables of $\varphi$. For all $i$ from 1 to $k$ we make the following case
    distinction:

    \begin{enumerate}[$\bullet$]
    \item If $\varphi_i=y$ for a variable $y$, then let $f_i(\varphi) = x_i \leftrightarrow y$.
    \item If $\varphi_i=\X \varphi_j$, then let $f_i(\varphi) = x_i \leftrightarrow \X x_j$.
    \item If $\varphi_i=\F \varphi_j$, then let $f_i(\varphi) = x_i \leftrightarrow \F x_j$.
    \item If $\varphi_i=\G \varphi_j$, then let $f_i(\varphi) = x_i \leftrightarrow \G x_j$.
    \item If $\varphi_i=\varphi_j \U \varphi_\ell$, then let $f_i(\varphi) = x_i \leftrightarrow x_j \U x_\ell$.
    \item If $\varphi_i=\varphi_j \S \varphi_\ell$, then let $f_i(\varphi) = x_i \leftrightarrow x_j \S x_\ell$.
    \item If $\varphi_i=g(\varphi_{i_1},\dots, \varphi_{i_n})$ for some $g \in B$,
          then let $f_i(\varphi) = x_i \leftrightarrow h(x_{i_1},\dots,x_{i_n})$,
          where $h$ is a formula using only $\wedge$, $\vee$, and $\neg$, representing the function $g$.
    \end{enumerate}

    Such a formula $h$ always exists with constant length, because the set $B$ is fixed and does not depend on the input.
    Now let $f(\varphi)=x_1 \wedge \bigwedge_{i=1}^k (\G f_i(\varphi) \wedge \neg (\mathop{true} \S \neg f_i(\varphi)))$ for case (1) and
    $f(\varphi)=x_1 \wedge \bigwedge_{i=1}^k \G f_i(\varphi)$ for case (2). The part $\G f_i(\varphi)$ makes sure that $f_i(\varphi)$ holds in every
    future state of the structure and $\neg (\mathop{true} \S \neg f_i(\varphi)))$ does the same for the past states of the structure.
    Additionally we consider $x \leftrightarrow y$ as a shorthand for $(x \wedge y) \vee (\neg x \wedge \neg y)$.
For case (1) we consider $\F x$ as a shorthand for $\mathop{true} \U x$ and $\G x$
    as a shorthand for $\neg(\mathop{true} \U \neg x)$, and for case (2) we consider $\G x$ as a shorthand for $\neg\F\neg x$. Thus we have that $f(\varphi)$ is from $\LL{\{\U,\S,\X\},\{\wedge,\vee,\neg\}}$ in case (1) and from $\LL{\{\F\},\{\wedge,\vee,\neg\}}$ in case (2).
%
%
 Furthermore $f$ is computable in logarithmic space, because the length of $f_i$ is polynomial
    and neither $\leftrightarrow$ nor the formulae $h$ occur nested. In order to show that
    $f$ is the reduction we are looking for, we still need to prove that $\varphi$ is satisfiable if and only if
    $f(\varphi)$ is satisfiable. Assume an arbitrary structure $S$, such that
    $S,s_i \vDash f(\varphi)$ for some $s_i$. We first prove by induction on the structure of the formula that $x_i$ holds if and only if $\varphi_i$ holds
    in every state $s$ of $S$ (for (1)) respectively in every state which lies in the future of $s_i$ (for (2)).
    Therefore for (1) let $s$ be an arbitrary state and for (2) let $s$ be an arbitrary state in the future of $s_i$. Thus by construction
    of $f(\varphi)$ the formulae $f_p(\varphi)$ hold at $s$ for all $1 \leq p \leq k$. Then the following holds:

    \begin{enumerate}[$\bullet$]
    \item
      If $\varphi_p=y$ for a variable $y$, then $f_p(\varphi) = x_p \leftrightarrow y$ and trivially $S, s \vDash x_p$
      iff $S, s \vDash y$.
    \item
      If $\varphi_p=\X\varphi_j$, then $f_p(\varphi) = x_p \leftrightarrow \X x_j$. Thus $S,s \vDash x_p$ iff for the
      successor state $s'$ of $s$, we have $S,s' \vDash x_j$. By induction this is equivalent to $S,s' \vDash \varphi_j$ and
      therefore $S,s \vDash \varphi_p$ iff $S,s \vDash x_p$.
    \item
      The cases for the temporal operator $\F$ or $\G$ work analogously.
    \item
      If $\varphi_p = \varphi_j \U \varphi_\ell$, then $f_p(\varphi) = x_p \leftrightarrow x_j \U x_\ell$. Thus $S,s \vDash x_p$
      iff there exists a state $s'$ in the future of $s$, such that $S,s' \vDash x_\ell$ and in all states $s_m$ in between
      (including $s$) $S,s_m \vDash x_j$. By induction this is equivalent to $S,s' \vDash \varphi_\ell$ and for all states
      in between $S,s_m \vDash \varphi_j$ and therefore $S,s \vDash \varphi_p$ iff $S,s \vDash x_p$.
    \item  \ifreport
      If $\varphi_p = \varphi_j \S \varphi_\ell$, then $f_p(\varphi) = x_p \leftrightarrow x_j \S x_\ell$. Thus $S,s \vDash x_p$
      iff there exists a state $s'$ in the past of $s$, such that $S,s' \vDash x_\ell$ and in all states $s_m$ in between
      (including $s$) $S,s_m \vDash x_j$. By induction this is equivalent to $S,s' \vDash \varphi_\ell$ and for all states
      in between $S,s_m \vDash \varphi_j$ and therefore $S,s \vDash \varphi_p$ iff $S,s \vDash x_p$.
     \else The case for the temporal operator \S\ works analogously to \U. \fi
    \item
      If $\varphi_p=g(\varphi_{i_1},\dots, \varphi_{i_n})$, then $f_p(\varphi) = x_p \leftrightarrow h(x_{i_1},\dots,x_{i_n})$,
      where $h$ is a formula using only $\wedge$, $\vee$, and $\neg$, representing the function $g$. Thus
      $S,s \vDash x_p$ iff $S,s \vDash h(x_{i_1},\dots,x_{i_n})$. Let $I$ be the subset of $I^n=\{i_1,\dots,i_n\}$, such that
      $S,s \vDash x_m$ for all $m \in I$ and $S,s \vDash \neg x_m$ for all $m \in I^n \setminus I$.
      By induction $S,s \vDash \varphi_m$ for all $m \in I$ and $S,s \vDash \neg \varphi_m$ for all
      $m \in I^n \setminus I$ and therefore $S,s \vDash h(\varphi_{i_1},\dots,\varphi_{i_n})$. Since
      $h$ represents the function $g$, we have that $S,s \vDash \varphi_p$ iff $S,s \vDash x_p$.
    \end{enumerate}

    Now, assume that $f(\varphi)$ is satisfiable. Then there exists a structure
    $S, s_i \vDash f(\varphi)$ and thus $S, s_i \vDash x_1$. Since in every state $x_j$ holds if and only if $\varphi_j$ holds,
    we have that $S, s_i \vDash \varphi = \varphi_1$. For the other direction, assume that $\varphi$ is satisfiable.
    Then there exists a structure $S,s_i \vDash \varphi=\varphi_1$. Now we can extend $S$ by adding new variables
    $x_1,\dots,x_k$ in such a way, that $x_j$ holds in a state $s$ from $S$ if and only if $\varphi_j$ holds in that state.
    Call this new structure $S'$. Then by construction of $f(\varphi)$, we have $S',s_i \vDash f(\varphi)$, since in every state
    $x_j$ holds if and only if $\varphi_j$ holds.
    \ifreport
      This concludes the proof of the first two cases.

      \bigskip\noindent
    \newcommand{\xdepth}[1]{\mathtext{depth}_{\X}\!\left(#1\right)}%
    We now show (3). For a formula $\varphi$ in which $\X$ is the only temporal operator, let $\xdepth\varphi$ denote the maximal nesting degree of the $\X$-operator in $\varphi,$ which we call the $\X$-depth of $\varphi.$ It is obvious that this number is linear in the length of $\varphi.$ Therefore, to show that the problem can be solved in \NP, it suffices to prove the following:

    \begin{enumerate}[(a)]
	\item Such a formula $\varphi$ is satisfiable if and only if there is a structure $S$ with the sequence $(s_i)_{i\in\mathbb N}$ such that for every $i>\xdepth\varphi,$ every variable in $s_i$ is false, and $S,s_0\models\varphi.$
        \item Given the assignments to the variables in the first $\xdepth\varphi$ states in the structure above, it can be verified in polynomial time if $S,s_0\models\varphi.$
    \end{enumerate}

    These claims immediately imply the complexity result. For the first point, it obviously suffices to show one direction. Therefore, let $S$ be an arbitrary structure with sequence $(s_i)_{i\in\mathbb{N}}$ such that $S,s_0\models\varphi,$ and let $S'$ be the structure with sequence $(s'_i)_{i\in\mathbb N}$ obtained from $S$ as follows: For $i\leq\xdepth\varphi,$ the assignment of the variables in the state $s'_i$ is the same as in $s_i.$ For $i>\xdepth\varphi,$ every variable is false in $s'_i.$ To prove claim (a) above, it suffices to prove that $S',s'_0\models\varphi.$

    To show this, we prove that for every subformula $\psi$ of $\varphi$ and every $i\leq\xdepth\varphi,$ if $\xdepth\psi\leq\xdepth\varphi-i,$ then $S,s_i\models\psi$ if and only if $S',s'_i\models\psi.$ For $i=0$ and $\psi=\varphi,$ this implies the desired result $S',s'_0\models\varphi.$

    We show the claim by induction on the formula $\psi.$ If $\psi$ is a variable, then, by construction, $S',s'_i\models\psi$ if and only if $S,s_i\models\psi,$ since the truth assignments of $s'_i$ and $s_i$ are identical. Now let $\psi$ be of the form $f(\psi_1,\dots,\psi_n)$ for an $n$-ary function $f\in B.$ In this case, it immediately follows that $\xdepth\psi=\max\set{\xdepth{\psi_1},\dots,\xdepth{\psi_n}}.$  Because of the prerequisites, $\xdepth\psi\leq\xdepth\varphi-i,$ and hence we know that for each $j\in\set{1,\dots,n},$ it holds that $\xdepth{\psi_j}\leq\xdepth\varphi-i.$ Therefore, we can apply the induction hypothesis to all of the $\psi_j,$ and we know that $S,s_i\models\psi_j$ if and only if $S',s'_i\models\psi_j.$ This immediately implies that $S,s_i\models\psi$ if and only if $S',s'_i\models\psi,$ since $f$ is a Boolean function.

    Finally, let $\psi$ be of the form $\X\xi$ for some formula $\xi.$ Hence, $\xdepth\psi=\xdepth\xi+1.$ Since $\xdepth\psi\leq\xdepth\varphi-i,$ this implies that $\xdepth\xi\leq\xdepth\varphi-(i+1).$ Hence, we can apply the induction hypothesis, and conclude that $S,s_{i+1}\models\xi$ if and only if $S',s'_{i+1}\models\xi.$ This immediately implies that $S,s_i\models\psi$ if and only if $S',s'_i\models\psi,$ and hence concludes the induction and the proof of claim (a).

    For claim (b), assume that $\varphi$ and the truth assignments for the first $\xdepth\varphi$ states in the structure $S$ are given, where all variables are assumed to be false in all further states. We can now, for each subformula $\psi$ of $\varphi,$ mark those states $s_i$ (for $i\leq\xdepth\varphi$) in which $\psi$ holds. Starting with $j=0,$ consider the subformulae of $\X$-depth $j.$ The question if a formula of $\X$-depth $j$ holds at a given state can easily be decided when this is known for all formulae of lower $\X$-depth. For $j=0,$ this can be decided easily, since the subformulae of $\X$-depth $0$ are exactly the propositional subformulae, and for these, each state can be considered separately. Additionally, observe that in the structure $S,$ all states beyond the first $\xdepth\varphi$ states satisfy exactly the same set of subformulae of $\varphi,$ hence only $\xdepth\varphi+1$ many states need to be considered.\qed
    \fi

  The following two theorems show that the case in which our Boolean operators are able to express the function $x\wedge\overline{y},$
  leads to \PSPACE-complete problems in the same cases as for the full set of Boolean operators. This function
  already played an important role in the classification result from \cite{lew79}, where it also marked the
  ``jump'' in complexity from polynomial time to \NP-complete.

  \begin{theorem}\label{theorem:SAT(G,X|F,X|U,X;S1) PSPACE-h}
    Let $B$ be a finite set of Boolean functions such that $\cS_1\subseteq\clone B$. Then 
      $\tsat{\{\G,\X\},B}$ and $\tsat{\{\F,\X\},B}$ are \PSPACE-complete.
  \end{theorem}

  \proof
      {Since it is possible to express \F\ using \G\ and negation, Theorem~\ref{theorem:sicl85}
      implies that $\tsat{\{\G,\X\},\set{\wedge,\vee,\neg}}$ and $\tsat{\{\F,\X\},\set{\wedge,\vee,\neg}}$ are \PSPACE-hard.
      Now, let $\varphi$ be
      a formula in which only temporal operators $\G$ and $\X$, or $\F$ and $\X$, and
      the Boolean connectives $\wedge,\vee,$ and $\neg$  appear.
      Let $B'=B\cup\set{1}$. The complete structure of Post's lattice~\cite{bcrv03} shows that $\clone{B'}=\cBF$.
      Now we can rewrite $\varphi$ as a $B'$-formula with the
      same temporal operators appearing. Due to Lemma~\ref{lemma:short formulae}, we can express the crucial operators
      $\wedge,\vee,\neg$ with short $B'$-formulae, \ie, formulae in which every relevant variable occurs only once.
      Therefore, this transformation can be performed in polynomial time. Now, in the $B'$-representation of $\varphi$, we
      exchange every occurrence of $1$ with a new variable $t$, and call the result $\varphi'$, which is a $B$-formula. It
      is obvious that $\varphi$ is satisfiable if and only if the $B$-formula $\varphi'\wedge t\wedge\G t$ is. Since
      $B\supseteq\cS_1$, we can express the occurring conjunctions using operators from $B$ (since these are a constant
      number of conjunctions, we do not need to worry about needing long $B$-formulae to express conjunction). This finishes
      the proof for $\tsat{\{\G,\X\},B}$. For the problem $\tsat{\{\F,\X\},B}$, observe that the function
      $g(x,y)=x\wedge\overline{y}$ generates the clone $\cS_1$, and therefore there is some $B$-formula equivalent to $g$.
      Now observe that the formula $t\wedge\overline{\F(t\wedge\overline{\X t})}=g(t,\F(g(t, \X t)))$ is equivalent to $\G t.$ Since this formula is independent of the input formula $\varphi$, this can be computed in polynomial time, and therefore this formula can be used to express $\varphi'\wedge t\wedge\G t$ in the same way as in the first case. Additionally, observe that if the operator $\F$ appears in the original formula $\varphi,$ then a subformula $\F\psi$ can be expressed as $(1\U\psi).$ Hence we conclude from Theorem~\ref{theorem:sicl85}\ref{theorem:sicl85:SAT(F,X|U|U,S,X;BF) PSPACE-c} that $\tsat{\{\U,\X\},\cBF}$ is \PSPACE-complete.\qed
      }


The construction in the proof of Theorem~\ref{theorem:SAT(G,X|F,X|U,X;S1) PSPACE-h} does not seem to be applicable to the languages
with \U\ and/or \S, as it requires a way to express $\G t$ using these operators. Hence, proving the desired completeness result
requires significantly more work.
Note that the case where $B$ contains the usual operators \AND, \OR, and negation, has already been proved in~\cite{mar04}. Our construction 
shows that hardness already holds for a class of propositional operators with less expressive power.


\begin{theorem}\label{theorem:SAT(S;BF) SAT(S|U;S1) PSPACE-h}
        Let $B$ be a finite set of Boolean functions with $\cS_1 \subseteq \clone B$.
        Then $\tsat{\{\S\},B}$ and $\tsat{\{\U\},B}$ are \PSPACE-complete.
  \end{theorem}

  \proof
    Since membership for \PSPACE\ is shown in Lemma~\ref{lemma:PSPACE_ub_all} we only need to show hardness. To do this, we give a reduction from $\mathtext{QBF}$. The main idea is to construct a temporal $B$-formula that requires satisfying models to simulate, in a linear structure, the quantifier evaluation tree of a quantified Boolean formula. Once we have ensured that models for the formula in fact are of this structure, we can prove that the quantified formula evaluation problem reduces to $\tsat{\{\S\},B}$.

    First we prove an auxiliary proposition for formulae of a special form which we use as building blocks in the construction. Intuitively the claim states that, given some propositional formulae $\varphi_1,\dots,\varphi_n$ that are pairwise contradictory, we can express that a model has a subsequence of states such that $\varphi_i$ holds in the $i$-th of these states.

    We cannot enforce that the $i$-th state always satisfies the $i$-th formula, since the truth of an LTL-formula using only $\S$ as a temporal operator is invariant under transformations of models that simply repeat a state finitely many times in the sequence.

        \begin{claim}
          Let $\varphi_1,\dots,\varphi_n$ be satisfiable propositional formulae such that $\varphi_i\rightarrow\neg\varphi_j$ is valid for all $i,j\in\{1,\dots,n\}$ with $i\neq j$. Then the formula
          \begin{equation*}
            \begin{split}
              \varphi =\varphi_1
                 \wedge(\varphi_1\S(\varphi_2\S(\dots\S(\varphi_{n-1}\S\varphi_n)\dots)))
                 \wedge((\dots((\varphi_1\S\varphi_2)\S\varphi_3)\S\dots)\S\varphi_n)
            \end{split}
          \end{equation*}
          is satisfiable and every structure $S$ that satisfies $\varphi$ in a state $s_m$ fulfills the following property: there exist natural numbers $0=a_0<a_1<\dots<a_n\leq m+1$ such that $m-a_i< j\leq m-a_{i-1}$ implies $S,s_j\vDash\varphi_i$ for every $i\in\{1\dots,n\}$.
        \end{claim}
        \begin{proof}
          Clearly $\varphi$ is satisfiable: since all formulae $\varphi_i$ are satisfiable we can find a structure $S$ such that $S,s_i\vDash\varphi_{n-i}$ for all $i\in\{0,\dots,n-1\}$. 
          One can verify that $S$ satisfies $\varphi$ in $s_{n-1}$.

          Let $S$ be a structure that satisfies $\varphi$ in a state $s_m$. Since $\varphi_i\rightarrow\neg\varphi_j$ is valid for all $i,j\in\{1,\dots,n\}$ with $i\neq j$, in every state only one of the formulae $\varphi_i$ can be satisfied by $S$. Therefore and since $S,s_m\vDash \varphi_1\S(\varphi_2\S(\dots\S(\varphi_{n-1}\S\varphi_n)\dots))$ holds, there are natural numbers $0=a_0\leq a_1\leq\dots\leq a_{n-1}<a_n\leq m+1$ such that $m-a_i< l\leq m-a_{i-1}$ implies $S,s_l\vDash\varphi_i$ for every $i\in\{1\dots,n\}$. Since $S,s_m\vDash\varphi_1$, it holds that $a_1>0$. Because $S,s_m\vDash(\dots((\varphi_1\S\varphi_2)\S\varphi_3)\S\dots)\S\varphi_n$ we conclude that $a_1<\dots <a_{n-1}$, which proves the claim.
        \end{proof}
        Now we give the reduction from $\mathtext{QBF}$, which is \PSPACE-complete due to~\cite{sto77}, to $\tsat{\{\S\},B}$. Let
        $\psi=Q_1x_1\dots Q_nx_n\varphi$
        for some propositional $\{\wedge,\vee,\neg\}$-formula $\varphi$ with variables $x_1,\dots,x_n$ and for quantifiers $Q_1,\dots,Q_n\in\{\forall,\exists\}$.

        Let $I_{\forall}=\{p_1,\dots,p_k\}=\{i\mid Q_i=\forall\}$ and
        $I_{\exists}=\{q_1,\dots,q_l\}=\{i\mid Q_i=\exists\}$ such that $p_1<\dots<p_k$ and $q_1<\dots<q_l$.

        We construct a temporal formula $\psi'\in \LL{\{\S\},B}$ such that $\psi$ is valid if and only if $\psi'$ is satisfiable. Let $t_0,\dots,t_n,u_0,\dots,u_n$ be new variables. We start with defining some subformulae using propositional operators from $\set{\neg,\vee,\wedge}$, then we combine them to obtain $\psi'$, and afterwards turn $\psi'$ into a temporal $B$-formula.

        \begin{small}
          ~\par\vspace{-1.3\baselineskip}
          \begin{equation*}
              \alpha= u_0\wedge\overline{t_0}
                  \wedge(u_0\wedge\overline{t_0})\S((
                  \overline{u_0}\wedge\overline{t_0})\S(\overline{u_0}\wedge t_0)))
                  \wedge(((u_0\wedge\overline{t_0})\S(
                  \overline{u_0}\wedge\overline{t_0}))\S(\overline{u_0}\wedge t_0))
          \end{equation*}
          ~\par\vspace{-1.0\baselineskip}
          \begin{minipage}{.52\textwidth}
            \begin{equation*} \begin{split}
              \beta&^1[i] =\\ \ & (u_{i-1}\wedge\overline{t_{i-1}}\wedge u_i\wedge\overline{t_i}\wedge \overline{x_i})
            \S \\
            & \ \ \ \
            ((\overline{u_{i-1}}\wedge\overline{t_{i-1}}\wedge\overline{u_i}\wedge\overline{t_i}\wedge \overline{x_i})
            \S \\
            & \ \ \ \ \ \ \ \
            ((\overline{u_{i-1}}\wedge\overline{t_{i-1}}\wedge\overline{u_i}\wedge t_i\wedge \overline{x_i})
            \S \\
            & \ \ \ \ \ \ \ \ \ \ \ \
            ((\overline{u_{i-1}}\wedge\overline{t_{i-1}}\wedge u_i\wedge\overline{t_i}\wedge x_i)
            \S \\
            & \ \ \ \ \ \ \ \ \ \ \ \ \ \ \ \
            ((\overline{u_{i-1}}\wedge\overline{t_{i-1}}\wedge\overline{u_i}\wedge\overline{t_i}\wedge x_i)
            \S \\
            & \ \ \ \ \ \ \ \ \ \ \ \ \ \ \ \ \ \ \ \
            (\overline{u_{i-1}}\wedge t_{i-1}\wedge\overline{u_i}\wedge t_i\wedge x_i )))))
            \end{split}\end{equation*}
          \end{minipage}
          \hspace{-.12\textwidth}
          \begin{minipage}{.53\textwidth}
            \begin{equation*} \begin{split}
              \beta^2[i] &=\\ (((( &(u_{i-1}\wedge\overline{t_{i-1}}\wedge u_i \wedge\overline{t_i}\wedge \overline{x_i})
            \\
            & \ \ \ \ \S
            (\overline{u_{i-1}}\wedge\overline{t_{i-1}}\wedge\overline{u_i}\wedge\overline{t_i}\wedge \overline{x_i}))
            \\
            & \ \ \ \ \ \ \ \ \S
            (\overline{u_{i-1}}\wedge\overline{t_{i-1}}\wedge\overline{u_i}\wedge t_i\wedge \overline{x_i}))
            \\
            & \ \ \ \ \ \ \ \ \ \ \ \ \S
            (\overline{u_{i-1}}\wedge\overline{t_{i-1}}\wedge u_i\wedge\overline{t_i}\wedge x_i ))
            \\
            & \ \ \ \ \ \ \ \ \ \ \ \ \ \ \ \ \S
            (\overline{u_{i-1}}\wedge\overline{t_{i-1}}\wedge\overline{u_i}\wedge\overline{t_i}\wedge x_i))
            \\
            & \ \ \ \ \ \ \ \ \ \ \ \ \ \ \ \ \ \ \ \ \S
            (\overline{u_{i-1}}\wedge t_{i-1}\wedge\overline{u_i}\wedge t_i\wedge x_i)
            \end{split}\end{equation*}
          \end{minipage}
          ~\par\vspace{.3\baselineskip}
          \begin{minipage}{.48\textwidth}
            \begin{equation*} \begin{split}
              \gamma^1[i] = \ & (u_{i-1}\wedge\overline{t_{i-1}}\wedge u_i\wedge\overline{t_i}\wedge\overline{x_i})
            \S \\
            & \ \ \ \
            ((\overline{u_{i-1}}\wedge\overline{t_{i-1}}\wedge\overline{u_i}\wedge\overline{t_i}\wedge\overline{x_i})
            \S \\
            & \ \ \ \ \ \ \ \
            ((\overline{u_{i-1}}\wedge t_{i-1}\wedge\overline{u_i}\wedge t_i\wedge\overline{x_i})))
            \end{split}\end{equation*}
          \end{minipage}
          \hspace{-.02\textwidth}
          \begin{minipage}{.48\textwidth}
            \begin{equation*} \begin{split}
              \gamma^2[i] = \ & (u_{i-1}\wedge\overline{t_{i-1}}\wedge u_i\wedge\overline{t_i}\wedge x_i)
            \S \\
            & \ \ \ \
            ((\overline{u_{i-1}}\wedge\overline{t_{i-1}}\wedge\overline{u_i}\wedge\overline{t_i}\wedge x_i)
            \S \\
            & \ \ \ \ \ \ \ \
            ((\overline{u_{i-1}}\wedge t_{i-1}\wedge\overline{u_i}\wedge t_i\wedge x_i)))
            \end{split}\end{equation*}
          \end{minipage}
        \par\bigskip
        \end{small}

        The formula $\alpha$ initializes a model as follows: it sets $u_0\overline{t_0}$ in the current state and requires that in the past there is a state with $\overline{u_0}t_0$ and all states in between satisfy $\overline{u_0}\overline{t_0}$. We will use $\beta^1[i]$ and $\beta^2[i]$ for $\forall$-quantified variables $x_i$ to partition the states such that $x_i$ is true in one partition and false in the other. Finally, we need $\gamma^1[i]$ and $\gamma^2[i]$ to set the values for the $\exists$-quantified variables.

        We now define the formula $\psi'$, which constitutes the reduction.

        $$\psi'= \alpha \;
            \wedge \displaystyle \bigwedge_{i\in I_\forall}\!((\beta^1[i] \wedge \beta^2[i])\S\, t_0)\;
            \wedge \displaystyle \bigwedge_{i\in I_\exists}\!((\gamma^1[i] \vee \gamma^2[i])\S\, t_0)\;
            \wedge \,(\varphi \S\, t_0)$$

        The formula $\psi'$ as defined above is specified as a formula using the connectives \AND, \OR, and \NOT. Before proving the correctness of the reduction, we show how $\psi'$ can be rewritten using only the available connectives from $B$. Due to the prerequisites, we know that $\cS_1\subseteq\clone B$. From the complete structure of Post's lattice~\cite{bcrv03}, it follows that $\clone{B\cup\{1\}}=\cBF$. Let $B'$ denote the set $B\cup\{1\}$. Since, due to Lemma~\ref{lemma:short formulae}, conjunction, disjunction, and negation can be written as $B'$-formulae such that every relevant variable appears only once, we can rewrite $\psi'$ into a temporal $B'$-formula with the result growing only polynomially in size (and the transformation can be carried out in polynomial time). Hence we can regard $\psi'$ as a temporal $B'$-formula. Now, since $\clone B\supseteq\cS_1$, and the \AND-function is an element of $\cS_1$, there is a $B$-formula $\mathtext{and}_B(x,y)$ which is equivalent to $x\wedge y$ (but both $x$ and $y$ might occur more than once in $\mathtext{and}_B(x,y)$). Now consider the propositional conjunctions of up to $5$ literals occurring in the subformulae $\beta^j[i]$, $\gamma^j[i]$, and $\alpha$ of $\psi'$, and recall that in the above step, we have rewritten these into formulae that only use connectives from $B$ and the constant $1$. For each such conjunction $\psi_{\mathtext{lit}}$, let $\psi^t_{\mathtext{lit}}$ be the formula obtained from $\psi_{\mathtext{lit}}$ by exchanging each occurrence of the constant $1$ with the new variable $t$. Now the formula $\mathtext{and}_B(t,\psi^t_{\mathtext{lit}})$ is equivalent to $\psi_{\mathtext{lit}}\wedge t$. We can therefore replace all formulae $\psi_{\mathtext{lit}}$ with $\mathtext{and}_B(t,\psi^t_{\mathtext{lit}})$, and obtain a formula which is equivalent to $\psi'$, but additionally forces the new variable $t$ to true in all the affected states. The remaining conjunctions occurring in the subformula $\alpha$ can simply be rewritten using the $\mathtext{and}_B(x,y)$-formula---there is only a constant number of these, hence this rewriting can be done in polynomial time.

        It remains to deal with conjunctions on the outmost level of $\psi'$, i.e., the three conjunctions connecting the different parts of the formula and the conjunctions over all ${i\in I_\forall}$ and ${i\in I_\exists}$. We first re-arrange these conjunctions as a formula which is a binary tree of logarithmic depth. Then each conjunction can be replaced by using the formula $\mathtext{and}_B(x,y)$ defined above. Since the nesting degree of the conjunction (and hence of applications of $\mathtext{and}_B(x,y)$) is only logarithmic, this transformation leads to a formula which is polynomial in the length of the original representation of $\psi'$, and can be carried out in polynomial time.

        The result of these transformations is a temporal $B$-formula which is equivalent to $\psi'$, apart from forcing the newly-introduced variable $t$ to true in all worlds in all models of $\psi'$ that lie in the scope of the relevant temporal operators. In particular, this formula is satisfiability-equivalent to $\psi'$. Hence it suffices to prove that the reduction is correct with respect to $\psi'$, i.e., that $\psi'$ is satisfiable if and only if the original $\mathtext{QBF}$-instance $\psi$ evaluates to true. For this, we first give a characterization of models satisfying $\psi'$, which establishes that models for this formula are indeed ``flat versions of quantifier-trees.''

        Hence assume that $S$ is a structure that satisfies $\psi'$ in a state $s_m$. We prove by induction over $n$ that there are natural numbers $0=a_0<\dots<a_{3(2^k)}\leq m+1$ and for every $q\in I_{\exists}$ a function $\sigma_q:\{0,1\}^{q-1}\rightarrow\{0,1\}$ such that $S$ satisfies the following property: if $m-a_i<j\leq m-a_{i-1}$, then it holds for all $h$ that

        \begin{enumerate}[(1)]
          \item\label{xp} $S,s_j\vDash x_{p_h}$  iff  $\lceil\frac{i}{3(2^{k-h})}\rceil$ is even,
          \item\label{xq} $S,s_j\vDash x_{q_h}$  iff  $\sigma_{q_h}(a_1\dots,a_{q_h-1})=1$ where $a_d=1$ if $x_d\in\xi(s_j)$ and $a_d=0$ otherwise,
          \item $S,s_j\vDash t_0$      iff  $i=3(2^k)$,
          \item $S,s_j\vDash t_{p_h}$  iff  $i=c\cdot 3(2^{k-h})$ for some $c\in\mathbb{N}$,\
          \item $S,s_j\vDash t_{q_h}$  iff  $S,s_j\vDash t_{p_h-1}$,
          \item $S,s_j\vDash u_0$      iff  $i=1$,
          \item $S,s_j\vDash u_{p_h}$  iff  $i=c\cdot 3(2^{k-h})+1$ for some $c\in\mathbb{N}$,
          \item $S,s_j\vDash u_{q_h}$  iff  $S,s_j\vDash u_{p_h-1}$.
        \end{enumerate}

	Note that due to point \ref{xp} for every possible assignment $\pi$ to $\{x_{p_1},\dots,x_{p_k}\}$ there is a $j\in\{m\!-\!a_{3(2^k)}\!+\!1,\dots,m\}$ such that $S,s_j\vDash x_{p_i}$ if and only if $\pi(x_{p_i})=1$. This is the main feature of the construction. The other variables $t_i$ and $u_i$ are necessary to ensure this condition. Figure~\ref{figure: structure} depicts the buildup of structures resulting from these eight properties. The states shown are necessary in a model for $\psi'$, however there can be more states in between but those have the same assignment as one of the displayed states. The assignment for the $\forall$-quantified variables $x_{p_1},\dots,x_{p_k}$ is given for all states and one can see that all possible assignments are present. Assignments to the $\exists$-quantified variables are not displayed because they can differ from structure to structure. The variables $u_i,t_i$ label all states which set them to true.

\begin{figure}
\begin{center}
  \ifpdf
    \includegraphics{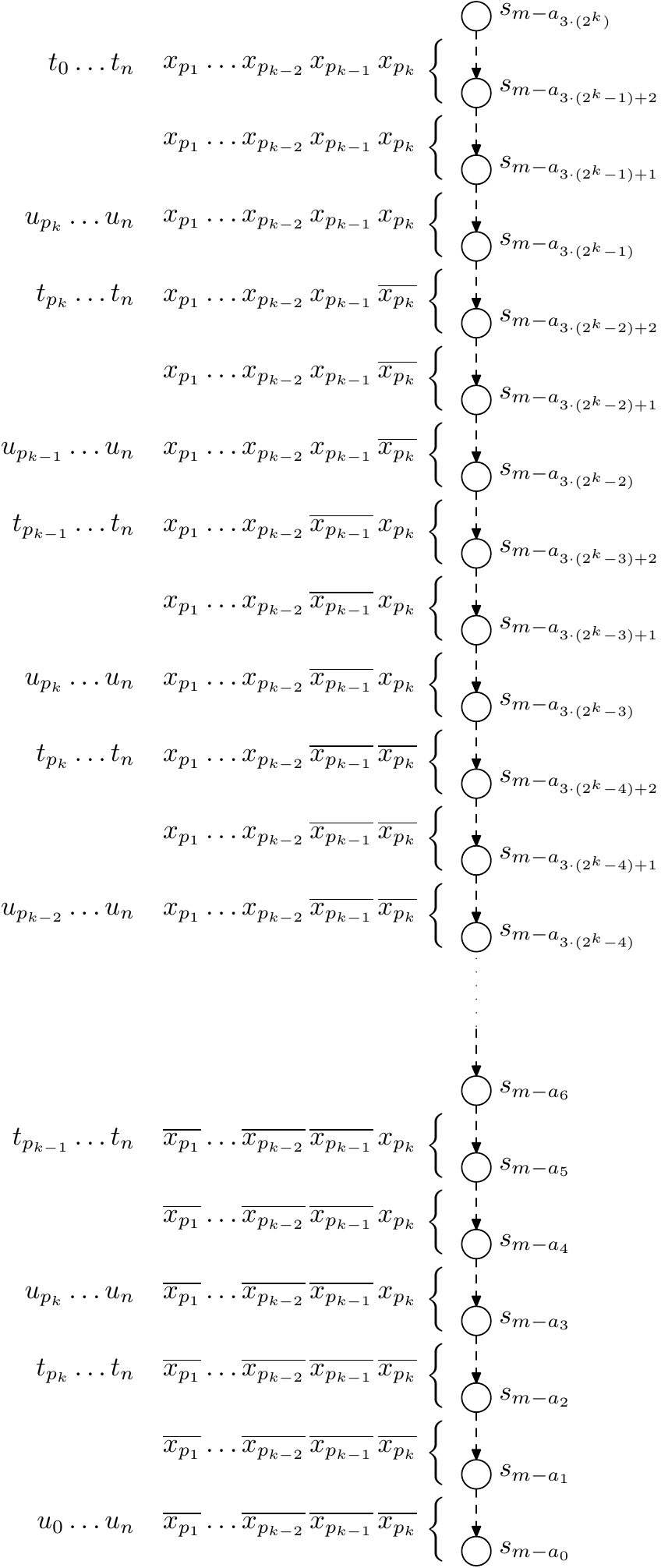}
  \else
    \includegraphics{Graphics/grafik_ilka.1}
  \fi
\caption{Structure of models of $\psi'$ in the proof of Theorem~\ref{theorem:SAT(S;BF) SAT(S|U;S1) PSPACE-h}}\label{figure: structure}
\end{center}
\end{figure}

        For $n=0$ it holds that $\psi'=\alpha\wedge(\varphi\S \,t_0)$. Since $\alpha$ satisfies the prerequisites of the claim above, there exist natural numbers $0=a_0<a_1<a_2<a_3\leq m+1$ such that
        \begin{enumerate}[$\bullet$]
          \item $m-a_1< j\leq m-a_0$ implies $S,s_j\vDash u_0\wedge\overline{t_0}$
          \item $m-a_2< j\leq m-a_1$ implies $S,s_j\vDash \overline{u_0}\wedge\overline{t_0}$
	  \item $m-a_3< j\leq m-a_2$ implies $S,s_j\vDash \overline{u_0}\wedge t_0$
        \end{enumerate}
        The only occurring variables are $u_0$ and $t_0$ and it is easy to see that the above property of $S$ holds for both.

        For the induction step assume that $n>1$ and the claim holds for $n-1$. There are two cases to consider:

        \begin{EcoEnum}
          \item[\textbf{Case 1:}] $Q_n=\forall$. That means
            \begin{equation*}\begin{split}
              \psi'=& \ \alpha\wedge\bigwedge_{i\in I_{\forall}\setminus\{n\}}((\beta^1[i]\wedge\beta^2[i])\S\, t_0)\wedge \bigwedge_{i\in I_{\exists}}((\gamma^1[i]\vee\gamma^2[i])\S\, t_0)\wedge (\varphi \S\, t_0) \\ &\wedge ((\beta^1[n]\wedge\beta^2[n]) \S\, t_0)
            \end{split}\end{equation*}
            It follows that there are natural numbers $0=a_0<\dots<a_{3(2^{k-1})}\leq m+1$ and for every $q\in I_{\exists}$ a function $\sigma_q:\{0,1\}^{q-1}\rightarrow\{0,1\}$ such that $S$ fulfills the properties of the claim (note that the subformula $(\varphi\S \,t_0)$ is not necessary for our argument). Since $S,s_m\vDash (\beta^1[n]\wedge\beta^2[n]) \S \,t_0$ and for $m-a_{3(2^{k-1})}< j\leq m$ it holds that $S,s_j\vDash t_0$ if and only if $j\leq m-a_{3(2^{k-1})-1}$, we have $S,s_j\vDash \beta^1[n]\wedge\beta^2[n]$ for every $m-a_{3(2^{k-1})-1}< j\leq m$. Let $i=c\cdot3$ for some $c\in\mathbb{N}$, then it holds that $m-a_{i+1}<j\leq m-a_{i}$ implies $S,s_j\vDash u_{n-1}$ which means that for these states $s_j$ it holds that $S,s_j\vDash u_{n-1}\wedge \overline{t_{n-1}}\wedge u_n\wedge \overline{t_n} \wedge x_n$. Due to our proposition there are natural numbers $0=b^i_0<b^i_1<\dots<b^i_6\leq a_i+1$ such that
            \begin{enumerate}[$\bullet$]
              \item $a_{i}-b^i_1< j\leq a_{i}-b^i_0$ implies $S,s_j\vDash u_{n-1}\wedge\overline{t_{n-1}}\wedge u_n\wedge\overline{t_n}\wedge \overline{x_n}$
              \item $a_{i}-b^i_2< j\leq a_{i}-b^i_1$ implies $S,s_j\vDash \overline{u_{n-1}}\wedge\overline{t_{n-1}}\wedge\overline{u_n}\wedge\overline{t_n}\wedge \overline{x_n}$
	      \item $a_{i}-b^i_3< j\leq a_{i}-b^i_2$ implies $S,s_j\vDash \overline{u_{n-1}}\wedge\overline{t_{n-1}}\wedge\overline{u_n}\wedge t_n\wedge \overline{x_n}$
              \item $a_{i}-b^i_4< j\leq a_{i}-b^i_3$ implies $S,s_j\vDash \overline{u_{n-1}}\wedge\overline{t_{n-1}}\wedge u_n\wedge\overline{t_n}\wedge x_n$
              \item $a_{i}-b^i_5< j\leq a_{i}-b^i_4$ implies $S,s_j\vDash \overline{u_{n-1}}\wedge\overline{t_{n-1}}\wedge\overline{u_n}\wedge\overline{t_n}\wedge x_n$
              \item $a_{i}-b^i_6< j\leq a_{i}-b^i_5$ implies $S,s_j\vDash \overline{u_{n-1}}\wedge t_{n-1}\wedge\overline{u_n}\wedge t_n\wedge x_n$
            \end{enumerate}
            The nearest state before $s_{m-a_{i}}$ that satisfies $\overline{u_{n-1}}$ is $s_{m-a_{i+1}}$ and the nearest state before $s_{m-a_{i}}$ that satisfies $t_{n-1}$ is $s_{m-a_{i+2}}$, therefore it holds that $b^i_1=a_{i+1}-a_{i}$ and $b^i_5=a_{i+2}-a_{i}$. By denoting $b^i_j+a_i$ with $c_{2i+j}$ we define natural numbers $c_0,\dots,c_{3(2^k)}$ for which it can be verified that they fulfill the claim.

          \item[\textbf{Case 2:}] $Q_n=\exists$. In this case we have
            \begin{equation*}\begin{split}
              \psi'=& \ \alpha\wedge\bigwedge_{i\in I_{\forall}}((\beta^1[i]\wedge\beta^2[i])\S\, t_0)\wedge \bigwedge_{i\in I_{\exists}\setminus\{n\}}((\gamma^1[i]\vee\gamma^2[i])\S\, t_0)\wedge (\varphi \S\, t_0) \\ &\wedge ((\gamma^1[n]\vee\gamma^2[n]) \S\, t_0).
            \end{split}\end{equation*}
          Because of the induction hypothesis there are natural numbers $0=a_0<a_1<\dots<a_{3(2^k)}\leq m+1$ such that the required properties are satisfied. Analogously to the first case $S,s_j\vDash\gamma^1[n]\vee\gamma^2[n]$ is true for every $m-a_{3(2^k)}<j\leq m$. Let $i=c\cdot3$, then for $m-a_{i+1}<j\leq m-a_i$ it holds that $S,s_j\vDash u_{n-1}\wedge\overline{t_{n-1}}\wedge u_n\wedge\overline{t_n}\wedge x_n$ or $S,s_j\vDash u_{n-1}\wedge\overline{t_{n-1}}\wedge u_n\wedge\overline{t_n}\wedge \overline{x_n}$, because $S,s_j\vDash u_{n-1}$. For $m-a_{i+2}<j\leq m-a_{i+1}$ we have that $S,s_j\vDash\overline{u_{n-1}}\wedge\overline{t_{n-1}}\wedge\overline{u_n}\wedge\overline{t_n}\wedge x_n$ or $S,s_j\vDash\overline{u_{i-n}}\wedge\overline{t_{i-n}}\wedge\overline{u_n}\wedge\overline{t_n}\wedge\overline{x_n}$ and for $m-a_{i+3}<j\leq m-a_{i+2}$ it must hold $S,s_j\vDash\overline{u_{n-1}}\wedge t_{n-1}\wedge\overline{u_n}\wedge t_n\wedge x_n$ or $S,s_j\vDash\overline{u_{n-1}}\wedge t_{n-1}\wedge\overline{u_n}\wedge t_n \wedge\overline{x_n}$. If $S,s_{a_i}\vDash\gamma^1[n]$, then in all these states $\overline{x_n}$ is satisfied; if $S,s_{a_i}\vDash\gamma^2[n]$, then $x_n$ is. Therefore with $\sigma_n$ defined by $\sigma_n(d_1,\dots, d_{n-1})=1$ if and only if $S,s_{3(d_12^{n-2}+\dots+d_{n-1}2^0)}\vDash\gamma^2[n]$, the induction is complete, because the binary numbers correspond to the assignments to the $\forall$-quantified variables.
        \end{EcoEnum}
        Note that for a structure that satisfies $\psi'$ with the above notation, $S,s_j\vDash\varphi$ holds for every $m-a_{3(2^k)}<j\leq m$, since $\varphi\S\, t_0$ is a conjunct of $\psi'$.

        Now assume that $\psi'$ is satisfiable in a state $s_m$ of a structure $S$. This is if and only if for every $q\in I_{\exists}$ there is a function $\sigma_{q}:\{0,1\}^{q-1}\rightarrow\{0,1\}$ such that $S$ fulfills the above property. Hence each possible assignment $J$ to the $\forall$-quantified variables $\{x_{p_1},\dots,x_{p_k}\}$ can be extended to an assignment to $\{x_1,\dots,x_n\}$ by $J(x_{q_i})=\sigma_{q_i}(J(x_1),\dots,J(x_{q_i-1}))$ which is equivalent to the validity of $\psi$.
        We can prove \PSPACE-hardness for $\tsat{\{\U\},B}$ with an analogous construction.\qed


  \ifreport
    In the following, we use the result from Lewis \cite{lew79} and the previously established upper
    bounds to obtain \NP-completeness results:
  \else
    The following proposition follows immediately from a result of Lewis's \cite{lew79}
    and the previously established upper bounds.
  \fi

  \begin{proposition}\label{prop:SAT(F|G|F,G|X;S1) NP-c}
    \ifreport
      Let $B$ be a finite set of Boolean functions such that $\cS_1\subseteq\clone B$.
    \else
      Let $B$ be a finite set of Boolean functions with $\cS_1\subseteq\clone B$.
    \fi
    Then \tsat{\{\F\},B}, \tsat{\{\G\},B}, \tsat{\{\F,\G\},B}, and \tsat{\{\X\},B} are \NP-complete.
  \end{proposition}

  \ifreport
    \proof
      Trivially, it holds that $\tsat{\emptyset,B} \redlogm \tsat{M,B}$ for each set $M$ of temporal operators,
      and \tsat{\emptyset,B} is \NP-complete due to \cite{lew79}. The upper bound
   follows from Theorem~\ref{theorem:sicl85}\ref{theorem:sicl85:SAT(F;BF) NP-c} and Lemma~\ref{lemma:PSPACE_ub_all}.\qed

  \fi
\subsection{Polynomial time results}

  \ifreport
    The following theorem shows that for some sets $B$ of Boolean functions, there is a satisfying model for every
    temporal $B$-formula over any set of temporal operators. These are the cases where $B\subseteq\cR_1,$ or
    $B\subseteq\cD.$ In the first case, every propositional formula over these operators is satisfied by the assignment
    giving the value $1$ to all appearing variables. In the second case, every propositional $B$-formula describes
    a self-dual function. For such a formula it holds in particular that if it is not satisfied by the all-zero
    assignment, then it is satisfied by the all-one assignment. Hence, such formulae are always satisfiable. It is easy
    to see that this is also true for temporal formulae involving these propositional operators.
  \else
    This subsection lists all cases for which LTL satisfiability can be decided in polynomial time.
    Due to the limitations of space, the proofs are omitted and can be found in the
    report~\cite{bsssv06}.

    As Theorem~\ref{theorem:SAT(.;R1|D) trivial} shows, for some sets $B$ of Boolean functions, there is a
    satisfying model for every temporal $B$-formula over any set of temporal operators.
  \fi

  \begin{theorem}\label{theorem:SAT(.;R1|D) trivial}
    \ifreport
      ~\par
      \begin{enumerate}[\em(1)]
        \item
          \label{it:SAT(.;R1) trivial}
          Let $B$ be a finite subset of $\cR_1$.
          Then every formula $\varphi$ from $\LL{\{\F,\!\!\;\G,\!\!\;\X,\!\!\;\U,\!\!\;\S\},B}$ is satisfiable.
        \item
          \label{it:SAT(.;D) trivial}
          Let $B$ be a finite subset of $\rm D$.
          Then every formula $\varphi$ from $\LL{\{\F,\G,\X,\U,\S\},B}$ is satisfiable.
      \end{enumerate}
    \else
      Let $B$ be a finite subset of $\cR_1$ or $\rm D$.
      Then every formula $\varphi$ from $\LL{\{\F,\!\!\;\G,\!\!\;\X,\!\!\;\U,\!\!\;\S\},B}$ is satisfiable.
    \fi
  \end{theorem}

  \ifreport
    \proof
      ~\par
      \begin{enumerate}[(1)]
        \item
          Since $\rm R_1$ is the class of 1-reproducing Boolean functions, any
          $\psi\in \rm R_1$ is true under the assignment that makes every propositional
          variable in $\psi$ true.
          If we apply this fact to formulae $\varphi \in \LL{\{\F,\G,\X,\U,\S\},B}$,
          then it is easy to see that any such formula $\varphi$ is true in
          every state of a structure $S_\varphi$ where the assignment of every state is $V_\varphi$.
        \item
        %
        We show by induction on the operators that this holds for all formulae. Let
        $S^1$ ($S^0$) denote the structure where the assignment of every
        state is $V_\varphi$ ($\emptyset$, resp.) and let $s^1$ ($s^0$, resp.) be the first state. We claim that
        $\varphi \in \LL{\{\F,\G,\X,\U,\S\},B}$ is satisfied by $S^1$ iff $\varphi$ is not satisfied
        by $S^0$. If $\varphi$ is purely propositional the claim holds trivially. We now have
        to look at the following cases:
        \begin{enumerate}[$\bullet$]
        \item
          $\varphi=\F \varphi_1$: Assume the claim holds for $\varphi_1$. Since for all states $s$ in $S^1$ the
          submodel starting at $s$ are isomorphic, obviously $\varphi$ is satisfied by $S^1$ iff $\F \varphi$ is
          satisfied by $S^1$ and the same argument also holds for $S^0$. Thus $S^0,s^0 \nvDash \varphi$ iff
          $S^1,s^1 \vDash \varphi$.
        \item
          $\varphi=\G \varphi_1$: This works analogously to $\F$.
        \item
          $\varphi=\X \varphi_1$: This also works analogously to $\F$.
        \item
          $\varphi=\varphi_1 \U \varphi_2$: Assume the claim holds for $\varphi_2$. Then
          $S^0,s^0 \nvDash \varphi$ iff $S^0,s^0 \nvDash \varphi_2$ iff $S^1,s^1 \vDash \varphi_2$
          iff $S^1,s^1 \vDash \varphi$.
        \item
          $\varphi=\varphi_1 \S \varphi_2$: This works analogously to $\U$.
        \item
          $\varphi=f(\varphi_1,\ldots,\varphi_n)$, such that $f$ is a self-dual function from $B$: Assume the claim holds
          for $\varphi_i$, $1\leq i \leq n$, \ie, $S^1,s^1\vDash \varphi_i$ iff
          $S^0,s^0\nvDash \varphi_i$. Then $S^1,s^1 \nvDash f(\varphi_1,\ldots,\varphi_n)$ implies
          $S^0,s^0 \vDash f(\varphi_1,\ldots,\varphi_n)$ and $S^1,s^1 \vDash f(\varphi_1,\ldots,\varphi_n)$ implies
          $S^0,s^0 \nvDash f(\varphi_1,\ldots,\varphi_n)$.\qed
        \end{enumerate}
      \end{enumerate}
  \fi

  \ifreport
    The following two theorems prove that satisfiability for formulae with any combination of
    modal operators, but only very restricted Boolean operators (\ie, negation and constants in
    the first case and only disjunction, conjunction, and constants in the second case), is always easy to decide.
  \else
    Due to Theorem~\ref{theorem:SAT(.;N) in P}, satisfiability for formulae with any combination of
    modal operators, but only very restricted Boolean operators is always easy to decide.
  \fi

  \ifreport
    \begin{theorem}\label{theorem:SAT(.;N) in P}
      Let $B$ be a finite subset of $\rm N$. Then $\tsat{\{\F,\G,\X,\U,\S\},B}$ can be decided in polynomial time.
    \end{theorem}

    \proof
      We give a recursive polynomial-time algorithm deciding the following question: Given a
      formula $\varphi$ built from propositional negation, constants, variables and arbitrary
      temporal operators, which of the following three cases occurs: $\varphi$ is unsatisfiable, $\varphi$ is a tautology, or $\varphi$ is not equivalent
      to a constant function. We also show that in the latter case, $\varphi$ is equivalent to a formula using only the above operators in which no constant appears. We will call these formulae temporal $\neg$-formulae.

      We give inductive criteria for these cases. Obviously, a constant $c$ is constant, and a variable is not, and can
      be written in the way defined above. The formula $\neg\varphi$ is equivalent to the constant $c$ if and only if
      $\varphi$ is equivalent to $\neg c$, otherwise it is equivalent to a temporal $\neg$-formula. If $\varphi=\F\varphi_1$,
      $\varphi=\G\varphi_1$, or $\varphi=\X\varphi_1$, then $\varphi$ is equivalent to a constant $c$ if and only if $\varphi_1$ is equivalent to $c:$ Obviously $\F c\equiv \G c\equiv \X c\equiv c$ for a constant. On the other hand, if $\varphi_1$ is not equivalent to a constant, then due to induction, it is equivalent to a temporal $\neg$-formula. Hence, $\F\varphi_1,$ $\G\varphi_1$ and $\X\varphi_1$ are equivalent to temporal $\neg$-formulae as well, and due to the proof of Theorem~\ref{theorem:SAT(.;R1|D) trivial}.~\ref{it:SAT(.;D) trivial}, these formulae are not equivalent to constants. Hence, if $\varphi_1$ is not equivalent to a constant, then $\varphi$ is not equivalent to a constant either, and can be written as a temporal $\neg$-formula.

      Now, let $\varphi=\varphi_1\U\varphi_2$. If $\varphi_2$ is a tautology, \ie, equivalent to the constant $1$, then, by
      the definition of $\U$, $\varphi$ is a tautology as well. Similarly, if $\varphi_2$ is equivalent to the constant $0$,
      then so is $\varphi$. Now assume that $\varphi_2$ is not constant. Then, by induction, $\varphi_2$ is equivalent to a
      temporal $\neg$-formula. If $\varphi_1$ is equivalent to the constant $0$, then $\varphi_1\U\varphi_2$ is equivalent
      to $\varphi_2$, and if $\varphi_1$ is equivalent to $1$, then $\varphi_1\U\varphi_2$ is equivalent to $\F\varphi_2$.
      If $\varphi_1$ is not equivalent to a constant, then, by induction, it can be written as a temporal $\neg$-formula, and
      obviously, this also holds for $\varphi_1\U\varphi_2$. Again due to the proof of
      Theorem~\ref{theorem:SAT(.;R1|D) trivial}.~\ref{it:SAT(.;D) trivial}, it follows that the entire formula $\varphi$ is not equivalent to a constant.

      For the operator $\S$, a similar argument can be made: Consider the formula $\varphi_1\S\varphi_2$. If $\varphi_2$ is
      a constant, then obviously the formula $\varphi_1\S\varphi_2$ is equivalent to the same constant. If $\varphi_1$ is
      the constant $0$, then $\varphi_1\S\varphi_2$ is equivalent to $\varphi_2$, and if $\varphi_1$ is the constant $1$, then
      $\varphi_1\S\varphi_2$ is equivalent to ``$\varphi_2$ was true at one point in the past.'' If $\varphi_2$ is not a
      constant, then this is equivalent to $\neg\varphi_2\S\varphi_2$, and thus this can be written as a temporal $\neg$-formula
      as well. As above, this formula is not equivalent to a constant. Now if both $\varphi_1$ and $\varphi_2$ are not
      equivalent to a constant function, then, by induction, both can be written as temporal $\neg$-formulae, and then
      $\varphi_1\S\varphi_2$ can be written as such a formula as well. In particular, with another application of the proof for Theorem~\ref{theorem:SAT(.;R1|D) trivial}~\ref{it:SAT(.;D) trivial}, $\varphi_1\S\varphi_2$ is not equivalent to a constant.

      This gives us a recursive algorithm deciding whether $\varphi$ is a constant, and if it is, which constant is equivalent to $\varphi$. The polynomial-time computable function $A_N$ is defined as follows: On input
      $\varphi$, $A_N(\varphi)=c\in\set{0,1}$ if $\varphi$ is equivalent to the constant $c$, and $A_N(\varphi)$ is the symbol
      {\ttfamily NOCONSTANT} if $\varphi$ is not equivalent to a constant.

      The function can be computed as follows: $A_N(c)$ is defined as $c$. For a variable $x$, $A_N(x)$ is the symbol
      {\ttfamily NOCONSTANT}. On input $\X\varphi$, $\G\varphi$, or $\F\varphi$, the algorithm returns $A_N(\varphi)$.
      On input $\varphi_1\U\varphi_2$, if $\varphi_2$ is a constant $c$, then $A_N(\varphi_1\U\varphi_2)=c$. Otherwise,
      if $\varphi_1$ is equivalent to $0$, then return $A_N(\varphi_2)$, and if $\varphi_1$ is equivalent to $1$, return
      $A_N(\F\varphi_2)$. If neither $\varphi_1$ nor $\varphi_2$ are constant, then return the symbol {\ttfamily NOCONSTANT}. Similarly, on input $\varphi_1\S\varphi_2$, if $\varphi_2$ is a constant $c$, then
      $A_N(\varphi_1\S\varphi_2)=c$. Otherwise, if $\varphi_1$ is the constant $0$, then
      $A_N(\varphi_1\S\varphi_2)=A_N(\varphi_2)$, and if $\varphi_1$ is the constant $1$, and $\varphi_2$ is not a constant,
      then $A_N(\varphi_1\S\varphi_2)$ is defined as the symbol {\ttfamily NOCONSTANT}. If $\varphi_1$ and $\varphi_2$ both are not a constant, then $A_N(\varphi_1\S\varphi_2)$ is again defined as the symbol {\ttfamily NOCONSTANT}.
      The function $A_N$ can obviously be computed in polynomial time, since there is at most
      one recursive call for each operator symbol in $\varphi.$

      By the argument above, this algorithm correctly determines if $\varphi$ is equivalent to the constant $0$ or the
      constant $1$. In particular, it determines if a given formula is satisfiable.\qed
      %

    \begin{theorem}\label{theorem:SAT(.;M) in P}
      Let $B$ be a finite subset of $\rm M$. Then $\tsat{\{\F,\G,\X,\U,\S\},B}$ can be decided in polynomial time.
    \end{theorem}

    \proof
      Remember that $\rm M$ is the clone of all monotone functions. 
      Let $\varphi$ be an arbitrary formula from \LL{\{\F,\G,\X,\U,\S\},B}. The following algorithm decides whether $\varphi$ is satisfiable.

      \newcommand{\MSAT}{\textsc{LTL-M-Sat}}
      \medskip\noindent
      \textbf{Algorithm} \MSAT
      \begin{algorithmic}
        \REPEAT
        \STATE{Replace all propositional sub-formulae that are unsatisfiable by 0}
        \STATE{Replace all sub-formulae $\F0$, $\G0$, $\X0$, $\psi\U 0$, $\psi\S 0$ by 0}
        \STATE{Replace all sub-formulae $0\U \psi$, $0\S \psi$ by $\psi$}
        \STATE{Replace all sub-formulae $\psi(\varphi_1,\dots,\varphi_k)$ by 0 if $\psi\in B$ and $\psi(\varphi'_1,\dots,\varphi'_k)$,
              where $\varphi'_i = 0$ if $\varphi_i=0$ and $\varphi'_i=1$ otherwise, is not true}
        \UNTIL{there are no changes anymore}
        \IF{$\varphi=0$}
        \STATE{{\bf return} ``unsatisfiable''}
        \ELSE
        \STATE{{\bf return} ``satisfiable''}
        \ENDIF
      \end{algorithmic}

      Since checking satisfiability of propositional $B$-formulae is in P (a $B$-formula $\varphi$ is satisfiable
      iff $\varphi(1,\dots,1)=1$) and there are at most as many replacements as there are operators in $\varphi$, \MSAT\
      runs in polynomial time.

      We prove that \MSAT\ is correct. If $\varphi$ is satisfiable, then \MSAT\ returns ``satisfiable.'' This is because
      all replacements in \MSAT\ do not affect satisfiability, so it follows that every formula \MSAT\ decides to be unsatisfiable
      is unsatisfiable. For the converse direction, let  $\varphi\in \LL{\{\F,\G,\X,\U,\S\},B}$ be such that \MSAT\ returns
      ``satisfiable'' and let $\varphi'$ be the formula generated by \MSAT\ in its REPEAT loop. We show by induction on the
      structure of $\varphi$ that $S,s_0 \vDash \varphi$, where $S=(s,V_{\varphi},\xi)$ is the structure in which every variable is true in every state, \ie, $\xi(s_i)=V_{\varphi}$ for every $i\in\mathbb{N}.$
      \begin{enumerate}[(1)]
      \item
        If $\varphi$ is a variable, it is satisfied in $S,s_0$ trivially.
      \item
        If $\varphi=\F \psi$ for a formula $\psi\in \LL{\{\F,\G,\X,\U,\S\},B}$, let $\psi'$ be the formula generated in the
        REPEAT loop when performing \MSAT\ on $\psi$. Assume that $\psi'=0$. Since every subformula replaced in $\psi$ by
        \MSAT\ will be replaced in $\varphi$, too, it holds that $\F \psi$ will be replaced by $\F0$ and that will be replaced
        by 0. It follows that $\varphi'=0$, but then \MSAT\ would return ``unsatisfiable.'' Thus, $\psi'\neq 0$, that means
        \MSAT\ returns ``satisfiable'' when performed on $\psi$. By induction it follows that $S,s_0 \vDash \psi$ and
        therefore $S,s_0 \vDash \varphi$ holds as well.
      \item
        If $\varphi=\G \psi$ for a formula $\psi\in \LL{\{\F,\G,\X,\U,\S\},B}$, we can use exactly the same arguments as in 2.
      \item
        If $\varphi=\X \psi$ for a formula $\psi\in \LL{\{\F,\G,\X,\U,\S\},B}$, we can use the same arguments as in 2.
      \item
        If $\varphi=\psi_1\U \psi_2$ for formulae $\psi_1,\psi_2\in \LL{\{\F,\G,\X,\U,\S\},B}$, we have that $\psi_2$ cannot
        be replaced by 0 (otherwise $\varphi$ would be replaced by 0 and \MSAT\ would return ``unsatisfiable''). So by
        induction it follows that $S,s_0\vDash \psi_2$. Hence, it holds that $S,s_0\vDash \varphi$ as well.
      \item
        If $\varphi=\psi_1\S \psi_2$ for formulae $\psi_1,\psi_2\in \LL{\{\F,\G,\X,\U,\S\},B}$, we can use the same arguments
        as for 5.
      \item
        If $\varphi=\psi(\varphi_1,\dots,\varphi_k)$ for formulae $\psi\in B$ and $\varphi_i \in \LL{\{\F,\G,\X,\U,\S\},B}$,
        for all $i = 1,\dots,k$, let $\varphi_1',\dots,\varphi_k'$ be the replacements of
        $\varphi_1,\dots,\varphi_k$. By induction it follows that $S,s_0\vDash \varphi_i$ if and only if $\varphi_i'\neq 0$
        for any $i\in\{1,\dots,k\}$. Since $\varphi'\neq 0$ and because of the last replacement rule, $S,s_0\vDash \varphi$.\qed
      \end{enumerate}
  \else
    \begin{theorem}\label{theorem:SAT(.;N) in P}
      Let $B$ be a finite subset of $\rm N$ or $\rm M$.
      Then $\tsat{\{\F,\G,\X,\U,\S\},B}$ can be decided in polynomial time.
    \end{theorem}
  \fi

  Finally, \ifreport we show that \fi satisfiability for formulae that have \X\ as a modal operator and the
  \XOR\ function $\oplus$ as a propositional operator is in \PTIME. This is true because functions described
  by these formulae have a high degree of symmetry.

  \begin{theorem}\label{theorem:SAT(X;L) in P}
    Let $B$ be a finite subset of $\cL.$ Then \tsat{\{\X\},B} can be decided in polynomial time.
  \end{theorem}

  \ifreport
    \proof
      First observe that any function from $\cL$ is of the form $f(x_1,\dots,x_n)=x_{i_1}\oplus\dots\oplus x_{i_k}\oplus c,$
      where the $x_{i_j}$ are pairwise different variables from the set $\set{x_1,\dots,x_n},$ and $c$ is either $0$ or $1.$
      Therefore, it is obvious that temporal $B$-formulae can be rewritten using only the connectors $\oplus$ and the constant $1$
      (the $0$ can be omitted in the representation above). Hence, we can assume that the set $B$ contains only
      the functions $\oplus$ and $1.$
      Now observe that any formula $\varphi$ from $\mathtext{L}(\{\X\},\{\oplus,1\})$ can be written as
      $$\varphi=\X\psi_1\oplus\dots\oplus\X\psi_k\oplus\psi,$$
      where $\psi$ is a propositional formula. This representation can be computed in polynomial
      time, and we can determine in polynomial time whether $\psi$ is a constant function.

      If $\psi$ is not a constant function, then $\varphi$ is satisfiable: Let
      $S=(s,V_{\varphi},\xi)$ be an arbitrary structure. If $\varphi$ is not satisfied at $s_0$, then
      we can ``switch over'' the current truth value of $\psi$, thus achieving that one more
      (or one less) of the arguments of the outermost \XOR\ function becomes true.
      For this purpose, we change the assignment of the propositional variables at $s_0$ in such
      a way that the new assignment satisfies $\psi$ if and only if the old assignment does not. Since this change
      does not affect the validity of the $\X\psi_i$ parts, $\varphi$ holds at $s_0$ with the new assignment.

      Now, if $\psi$ is constant, this trick does not work. Instead, let
      $$\varphi'=\psi_1\oplus\dots\oplus\psi_k.$$
      Observe that in this case $\varphi$ is satisfiable if and only if
      $\psi$ is the constant $0$ and $\varphi'$ is satisfiable, or if $\psi$ is the constant $1$ and $\varphi'$ is no
      tautology; and that $\varphi$ is a tautology if and only if $\psi$ is the constant $0$ and $\varphi'$ is a
      tautology, or $\psi$ is the constant $1$ and $\varphi'$ is not satisfiable.
      Thus we have an iterative algorithm deciding \tsat{\{\X\},\{\oplus,1\}}, since for a propositional $B$-formula, these
      questions can be efficiently decided.\qed
  \fi

\section{Conclusion}

  We have almost completely classified the computational complexity of satisfiability
  for LTL with respect to the sets of propositional and temporal operators
  permitted, see Table \ref{tab:results_concl}. The only case left open is the one in which only propositional
  operators constructed from the binary \XOR\ function (and, perhaps, constants)
  are allowed. This case has already turned out to be difficult to
  handle---and hence was left open---in~\cite{bhss06} for modal satisfiability
  under \textit{restricted} frames classes.
  The difficulty here and in~\cite{bhss06} is reflexivity, \ie, the property that
  the formula $\F\varphi$ is satisfied at some state if $\varphi$ is satisfied
  at \textit{the same} state. This does not allow for a separate treatment of
  the propositional part (without temporal operators) and the remainder of a
  given formula.
  \begin{table}
    \centering
    \begin{small}
      \begin{tabular}{l|cc}
        \hspace*{\fill} temporal operators               & $\{\F\}$, $\{\G\}$,   & any other         \\[2pt]
        function class $B$ (propositional operators)     & $\{\F,\G\}$, $\{\X\}$ & combination       \\
        \hline
        \Stab $B \subseteq \cR_1$ or $B \subseteq \cD$   & trivial               & trivial           \\
        \stab $B \subseteq \cM$ or $B \subseteq \cN$     & in \PTIME             & in \PTIME         \\
        \stab $\cL_0$, $\cL$                             & ?                     & ?                 \\
        \stab else (\ie, $B \supseteq \cS_1$)            & \NP-c.\               & \PSPACE-c.\
      \end{tabular}

    \end{small}
    \par\bigskip
    \caption{%
      Complexity results for satisfiability. The entries ``trivial'' denote cases in which a given formula
      is always satisfiable. The abbreviation ``c.'' stands for ``complete.'' Question marks stand for open questions.%
    }
    \label{tab:results_concl}
  \end{table}

  Our results bear an interesting resemblance to the classifications obtained
  in \cite{lew79} and in \cite{bhss06}. In all of these cases (except for one of the several
  classifications obtained in the latter), it turns out that sets of Boolean functions $B$
  which generate a clone above $\cS_1$ give rise to computationally hard problems,
  while other cases seem to be solvable in polynomial time. Therefore, in a precise sense, it is
  the function represented by the formula $x\wedge\overline{y}$ which turns problems in
  this context computationally intractable. 
%
  These hardness results seem to indicate that $x\wedge\overline{y}$ and other functions which
  generate clones above $\cS_1$ have properties that make computational problems hard, and this notion
  of hardness is to a large extent independent of the actual problem considered.

  In \cite{bms+07}, we have separated tractable and intractable cases of the model checking problem for LTL with restrictions 
  to the propositional operators. Without such restrictions, this problem has the same complexity as satisfiability~\cite{sicl85}.

  The results from this paper leave two open questions. Besides the unsolved \XOR\ case, it would be
  interesting to further classify the polynomial-time solvable cases.
  Further work could also examine related specification languages, such as CTL, $\text{CTL}^\ast$,
  or hybrid temporal languages.

\section*{Acknowledgments}
We thank Martin Mundhenk and the anonymous referees for helpful comments and suggestions.

  \bibliographystyle{alpha}
  \bibliography{\jobname}

\newcommand{\etalchar}[1]{$^{#1}$}
\begin{thebibliography}{BMS{\etalchar{+}}07}

\bibitem[BCRV03]{bcrv03}
E.~B\"ohler, N.~Creignou, S.~Reith, and H.~Vollmer.
\newblock Playing with {B}oolean blocks, part {I}: Post's lattice with
  applications to complexity theory.
\newblock {\em SIGACT News}, 34(4):38--52, 2003.

\bibitem[BHSS06]{bhss06}
M.~Bauland, E.~Hemaspaandra, H.~Schnoor, and I.~Schnoor.
\newblock Generalized modal satisfiability.
\newblock In B.~Durand and W.~Thomas, editors, {\em STACS}, volume 3884 of {\em
  Lecture Notes in Computer Science}, pages 500--511. Springer, 2006.

\bibitem[BMS{\etalchar{+}}07]{bms+07}
M.~Bauland, M.~Mundhenk, T.~Schneider, H.~Schnoor, I.~Schnoor, and H.~Vollmer.
\newblock The tractability of model checking for {LTL}: the good, the bad, and
  the ugly fragments.
\newblock In {\em Proceedings Methods for Modalities 5}, pages 125--140. ENS
  Cachan, 2007.
\newblock Also at CoRR http://arxiv.org/abs/0805.0498.

\bibitem[BSS{\etalchar{+}}07]{bss+07}
M.~Bauland, T.~Schneider, H.~Schnoor, I.~Schnoor, and H.~Vollmer.
\newblock The complexity of generalized satisfiability for linear temporal
  logic.
\newblock In H.~Seidl, editor, {\em FoSSaCS}, volume 4423 of {\em Lecture Notes
  in Computer Science}, pages 48--62. Springer, 2007.

\bibitem[CL93]{CL93}
C.-C Chen and I-P. Lin.
\newblock The computational complexity of satisfiability of temporal {Horn}
  formulas in propositional linear-time temporal logic.
\newblock {\em Inf. Process. Lett.}, 45(3):131--136, 1993.

\bibitem[Coo71]{coo71a}
S.~A. Cook.
\newblock The complexity of theorem proving procedures.
\newblock In {\em Proceedings 3rd Symposium on Theory of Computing}, pages
  151--158. ACM Press, 1971.

\bibitem[Dal00]{dal00}
V.~Dalmau.
\newblock {\em Computational Complexity of Problems over Generalized Formulas}.
\newblock PhD thesis, Department de Llenguatges i Sistemes Inform\`atica,
  Universitat Polit\'ecnica de Catalunya, 2000.

\bibitem[DFR00]{DFR00}
C.~Dixon, M.~Fisher, and M.~Reynolds.
\newblock Execution and proof in a {Horn}-clause temporal logic.
\newblock In H.~Barringer, M.~Fisher, D.~Gabbay, and G.~Gough, editors, {\em
  Advances in Temporal Logic}, volume~16 of {\em Applied Logic Series}, pages
  413--433. Kluwer, 2000.

\bibitem[DS02]{ds02}
S.~Demri and P.~Schnoebelen.
\newblock The complexity of propositional linear temporal logics in simple
  cases.
\newblock {\em Inf. Comput.}, 174(1):84--103, 2002.

\bibitem[EES90]{ees90}
E.~A. Emerson, M.~Evangelist, and J.~Srinivasan.
\newblock On the limits of efficient temporal decidability.
\newblock In {\em LICS}, pages 464--475. IEEE Computer Society, 1990.

\bibitem[Hal95]{hal95}
J.~Y. Halpern.
\newblock The effect of bounding the number of primitive propositions and the
  depth of nesting on the complexity of modal logic.
\newblock {\em Artif. Intell.}, 75(2):361--372, 1995.

\bibitem[Hem01]{hem01}
E.~Hemaspaandra.
\newblock The complexity of poor man's logic.
\newblock {\em J. Log. Comput.}, 11(4):609--622, 2001.

\bibitem[Lew79]{lew79}
H.~Lewis.
\newblock Satisfiability problems for propositional calculi.
\newblock {\em Mathematical Systems Theory}, 13:45--53, 1979.

\bibitem[Low08]{low08}
G.~Lowe.
\newblock Specification of communicating processes: temporal logic versus
  refusals-based refinement.
\newblock {\em Formal Aspects of Computing}, 20(3):277--294, 2008.

\bibitem[Mar04]{mar04}
N.~Markey.
\newblock Past is for free: on the complexity of verifying linear temporal
  properties with past.
\newblock {\em Acta Informatica}, 40(6-7):431--458, 2004.

\bibitem[Nor05]{nor05}
G.~Nordh.
\newblock A trichotomy in the complexity of propositional circumscription.
\newblock In {\em Proceedings of the 11th International Conference on Logic for
  Programming}, volume 3452 of {\em Lecture Notes in Computer Science}, pages
  257--269. Springer Verlag, 2005.

\bibitem[Pip97]{pip97b}
N.~Pippenger.
\newblock {\em Theories of Computability}.
\newblock Cambridge University Press, Cambridge, 1997.

\bibitem[Pnu77]{pnu77}
A.~Pnueli.
\newblock The temporal logic of programs.
\newblock In {\em FOCS}, pages 46--57. IEEE, 1977.

\bibitem[Pos41]{pos41}
E.~Post.
\newblock The two-valued iterative systems of mathematical logic.
\newblock {\em Annals of Mathematical Studies}, 5:1--122, 1941.

\bibitem[Rei01]{rei01}
S.~Reith.
\newblock {\em Generalized Satisfiability Problems}.
\newblock PhD thesis, Fachbereich Mathematik und Informatik, Universit\"at
  W\"urzburg, 2001.

\bibitem[RV03]{revo03}
S.~Reith and H.~Vollmer.
\newblock Optimal satisfiability for propositional calculi and constraint
  satisfaction problems.
\newblock {\em Information and Computation}, 186(1):1--19, 2003.

\bibitem[RW05]{rewa99-dt}
S.~Reith and K.~W. Wagner.
\newblock The complexity of problems defined by {B}oolean circuits.
\newblock In {\em Proceedings International Conference Mathematical Foundation
  of Informatics, (MFI99); World Science Publishing}, 2005.

\bibitem[SC85]{sicl85}
A.~Sistla and E.~Clarke.
\newblock The complexity of propositional linear temporal logics.
\newblock {\em Journal of the ACM}, 32(3):733--749, 1985.

\bibitem[Sch05]{sch05}
H.~Schnoor.
\newblock The complexity of the {B}oolean formula value problem.
\newblock Technical report, Theoretical Computer Science, University of
  Hannover, 2005.

\bibitem[Sto77]{sto77}
L.~Stockmeyer.
\newblock The polynomial-time hierarchy.
\newblock {\em Theoretical Computer Science}, 3:1--22, 1977.

\end{thebibliography}
\end{document}